\documentclass{article}
\usepackage{amsmath,amstext,amssymb,amsfonts}
\usepackage{color,graphicx}
\usepackage{tabularx}
\usepackage{verbatim}
\usepackage{pgf,tikz}
\usepackage{mathrsfs}
\usepackage{algorithm2e}
\usepackage{algorithmic}
\usetikzlibrary{arrows}
\usetikzlibrary{decorations.pathreplacing,angles,quotes}

\usepackage{amsthm}
\usepackage{ifdraft}

\usepackage{nicefrac}

\newtheorem{theorem}{Theorem}[section]
\newtheorem*{theorem*}{Theorem}

\newtheorem{Claim}[theorem]{Claim}
\newtheorem*{claim*}{Claim}

\newtheorem{proposition}[theorem]{Proposition}
\newtheorem*{proposition*}{Proposition}
\newtheorem{lemma}[theorem]{Lemma}
\newtheorem*{lemma*}{Lemma}

\newtheorem*{conjecture*}{Conjecture}

\newtheorem*{fact*}{Fact}

\newtheorem*{hypothesis*}{Hypothesis}

\theoremstyle{definition}
\newtheorem{definition}[theorem]{Definition}

\newtheorem{remark}[theorem]{Remark}

\usepackage{prettyref}
\newcommand{\savehyperref}[2]{\texorpdfstring{\hyperref[#1]{#2}}{#2}}

\newrefformat{eq}{\savehyperref{#1}{\textup{(\ref*{#1})}}}
\newrefformat{lem}{\savehyperref{#1}{Lemma~\ref*{#1}}}
\newrefformat{def}{\savehyperref{#1}{Definition~\ref*{#1}}}
\newrefformat{thm}{\savehyperref{#1}{Theorem~\ref*{#1}}}
\newrefformat{cor}{\savehyperref{#1}{Corollary~\ref*{#1}}}
\newrefformat{cha}{\savehyperref{#1}{Chapter~\ref*{#1}}}
\newrefformat{sec}{\savehyperref{#1}{Section~\ref*{#1}}}
\newrefformat{app}{\savehyperref{#1}{Appendix~\ref*{#1}}}
\newrefformat{tab}{\savehyperref{#1}{Table~\ref*{#1}}}
\newrefformat{fig}{\savehyperref{#1}{Figure~\ref*{#1}}}
\newrefformat{hyp}{\savehyperref{#1}{Hypothesis~\ref*{#1}}}
\newrefformat{alg}{\savehyperref{#1}{Algorithm~\ref*{#1}}}
\newrefformat{sdp}{\savehyperref{#1}{SDP~\ref*{#1}}}
\newrefformat{rem}{\savehyperref{#1}{Remark~\ref*{#1}}}
\newrefformat{item}{\savehyperref{#1}{Item~\ref*{#1}}}
\newrefformat{step}{\savehyperref{#1}{step~\ref*{#1}}}
\newrefformat{conj}{\savehyperref{#1}{Conjecture~\ref*{#1}}}
\newrefformat{fact}{\savehyperref{#1}{Fact~\ref*{#1}}}
\newrefformat{prop}{\savehyperref{#1}{Proposition~\ref*{#1}}}
\newrefformat{claim}{\savehyperref{#1}{Claim~\ref*{#1}}}
\newrefformat{relax}{\savehyperref{#1}{Relaxation~\ref*{#1}}}
\newrefformat{red}{\savehyperref{#1}{Reduction~\ref*{#1}}}
\newrefformat{part}{\savehyperref{#1}{Part~\ref*{#1}}}
\newrefformat{prob}{\savehyperref{#1}{Problem~\ref*{#1}}}
\newrefformat{ass}{\savehyperref{#1}{Assumption~\ref*{#1}}}

\newcommand{\Sref}[1]{\hyperref[#1]{\S\ref*{#1}}}

\usepackage[varg]{txfonts}    

\renewcommand{\mathbb}{\varmathbb} 

\renewcommand{\leq}{\leqslant}

\renewcommand{\geq}{\geqslant}
\renewcommand{\ge}{\geqslant}

\usepackage{bm}

\usepackage{xspace}

\usepackage[pdftex,pagebackref,colorlinks,linkcolor=blue,filecolor = blue, citecolor = blue, urlcolor  = blue]{hyperref}

\usepackage{boxedminipage}

\newcommand{\mper}{\,.}
\newcommand{\mcom}{\,,}

\newcommand{\paren}[1]{\left(#1 \right )}

\newcommand{\brac}[1]{[#1 ]}
\newcommand{\Brac}[1]{\left[#1\right]}

\newcommand{\set}[1]{\left\{#1\right\}}
\newcommand{\Set}[1]{\left\{#1\right\}}

\newcommand{\abs}[1]{\left\lvert#1\right\rvert}
\newcommand{\Abs}[1]{\left\lvert#1\right\rvert}

\newcommand{\norm}[1]{\left\lVert#1\right\rVert}

\newcommand{\defeq}{\stackrel{\textup{def}}{=}}

\newcommand{\inprod}[1]{\left\langle #1\right\rangle}

\newcommand{\Z}{{\mathbb Z}}
\newcommand{\N}{{\mathbb Z}_{\geq 0}}
\newcommand{\R}{\mathbb R}
\newcommand{\Q}{\mathbb Q}

\newcommand{\Esymb}{\mathbb{E}}
\newcommand{\Psymb}{\mathbb{P}}

\DeclareMathOperator*{\E}{\Esymb}

\DeclareMathOperator*{\ProbOp}{\Psymb}

\newcommand{\given}{\mathrel{}\middle|\mathrel{}}

\renewcommand{\Pr}[1]{\ProbOp\Brac{#1}}

\newcommand{\e}{\epsilon}

\definecolor{DSgray}{cmyk}{0,0,0,0.7}

\let\e\varepsilon

\newcommand{\cB}{\mathcal B}

\newcommand{\cF}{\mathcal F}

\newcommand{\cI}{\mathcal I}

\newcommand{\cL}{\mathcal L}

\newcommand{\cN}{\mathcal N}

\newcommand{\cR}{\mathcal R}

\newcommand{\cT}{\mathcal T}

\newcommand{\cV}{\mathcal V}

\newcommand{\bbC}{\mathbb C}

\newcommand{\bigO}{\mathcal{O}}
\newcommand{\bigo}[1]{\bigO\left(#1\right)}
\newcommand{\tbigO}{\tilde{\mathcal{O}}}
\newcommand{\tbigo}[1]{\tbigO\left(#1\right)}

\newcommand{\poly}{{\sf poly}}

\usepackage{full page}
\newcommand{\diam}[1]{~{\sf diam}\paren{#1}}
\newcommand{\radm}[1]{~{\sf rad_{min}}\paren{#1}}

\newcommand{\tmu}{\tilde{\mu}}
\newcommand{\tmuu}[2]{\tmu_{#1}^{\paren{#2}}}

\newcommand{\rij}{r_{i,j}}

\newcommand{\wmin}{\omega}
\newcommand{\simplex}{\mathcal V'}
\newcommand{\Opt}{C}

\title{On Euclidean $k$-Means Clustering with $\alpha$-Center Proximity }
\author{Amit Deshpande \\ Microsoft Research, India \\ \href{mailto:amitdesh@microsoft.com}{amitdesh@microsoft.com}
\and
Anand Louis \\ Indian Institute of Science \\ \href{mailto:anandl@iisc.ac.in}{anandl@iisc.ac.in}
\and
Apoorv Singh \\ Indian Institute of Science \\ \href{mailto:apoorvsingh@iisc.ac.in}{apoorvsingh@iisc.ac.in} }

\date{}
\begin{document}
\begin{titlepage}

\maketitle
\thispagestyle{empty}

\begin{abstract}
  $k$-means clustering is NP-hard in the worst case but previous work has shown efficient algorithms assuming the optimal $k$-means clusters are \emph{stable} under additive or multiplicative perturbation of data. This has two caveats. First, we do not know how to efficiently verify this property of optimal solutions that are NP-hard to compute in the first place. Second, the stability assumptions required for polynomial time $k$-means algorithms are often unreasonable when compared to the ground-truth clusters in real-world data. A consequence of multiplicative perturbation resilience is \emph{center proximity}, that is, every point is closer to the center of its own cluster than the center of any other cluster, by some multiplicative factor $\alpha > 1$.


  We study the problem of minimizing the Euclidean $k$-means objective only over clusterings that satisfy $\alpha$-center proximity. We give a simple algorithm to find the optimal $\alpha$-center-proximal $k$-means clustering in running time exponential in $k$ and $1/(\alpha - 1)$ but linear in the number of points and the dimension. We define an analogous $\alpha$-center proximity condition for outliers, and give similar algorithmic guarantees for $k$-means with outliers and $\alpha$-center proximity. On the hardness side we show that for any $\alpha' > 1$, there exists an $\alpha \leq \alpha'$, $(\alpha >1)$, and an $\e_0 > 0$ such that minimizing the $k$-means objective over clusterings that satisfy $\alpha$-center proximity is NP-hard to approximate within a multiplicative $(1+\e_0)$ factor.

\end{abstract}
\end{titlepage}

\newcommand{\inr}{\cR_1}
\newcommand{\outr}{\cR_2}
\newcommand{\rhoi}{\tilde{R}}
\newcommand{\suchthat}{\;:\;}

\section{Introduction}
Clustering is an important tool in any data science toolkit. Most popular clustering algorithms partition the given data into disjoint clusters by optimizing a certain global objective such as the $k$-means. The implicit assumption in doing so is that an optimal solution for this objective would recover the underlying \emph{ground truth} clustering. However, many such objectives are NP-hard to optimize in the worst case, e.g., $k$-center, $k$-median, $k$-means. Moreover, an optimal solution need not satisfy certain properties desired from the ground truth clustering, e.g., balance, stability. We highlight this problem using the example of $k$-means.

Given a set of $n$ points $X = \{x_{1}, x_{2}, \dotsc, x_{n}\}$ in a metric space with the underlying metric $\text{dist}(\cdot, \cdot)$, and a positive integer $k$, the $k$-means objective is to find centers $\mu_{1}, \mu_{2}, \dotsc, \mu_{k}$ in the given metric space so as to minimize the sum of squared distances of all the points to their nearest centers, respectively, i.e.,
minimize $\sum_{i=1}^{n} \min_{1 \leq j \leq k} \text{dist}(x_{i}, \mu_{j})^{2}$. This results in clusters $C_{1}, C_{2}, \dotsc, C_{k}$, where the cluster $C_{j}$ consists of all the points $x_{i}$ whose nearest center is $\mu_{j}$.
Optimization of the $k$-means objective is NP-hard in the worst case, even when for Euclidean $k$-means with $k=2$ \cite{AloiseDHP2009,DasguptaF2009} or $d=2$ \cite{MahajanNV2012}. There are known worst-case instances where the popular Lloyd's algorithm for $k$-means takes exponentially many iterations for convergence to a local optimum \cite{Vattani2011}.
On the algorithmic side, $k$-means++ initialization gives $O(\log k)$-approximation, and there are known $O(1)$-approximations in time $\text{poly}(n, k, d)$ \cite{JainVazirani} and $(1+\e)$-approximations for the Euclidean $k$-means in time $O(2^{\text{poly}(k/\e)} \cdot nd)$ \cite{KumarSS04,Chen09}.
Euclidean $k$-means is known to be NP-hard to approximate within some fixed constant $c > 1$ \cite{AwasthiCKS15}, hence, the exponential dependence on $k$ is necessary for $(1+\e)$-approximation.

In practice, however, Lloyd's algorithm, $k$-means++, and their variants perform well on most real-world data sets. This dichotomy between theoretical intractability and empirically observed efficiency has lead to the CDNM thesis \cite{Ben-David15}: Clustering is difficult only when it does not matter! In most real-world data sets, the underlying ground-truth clustering is unambiguous and stable under small perturbations of data. As a consequence, the ground-truth clustering satisfies \emph{center proximity}, that is, every point is closer to the center of its own cluster than the center of any other cluster, by some multiplicative factor $\alpha > 1$.
Thus, center proximity is a desirable property for the output of any clustering algorithm used in practice. Balanced clusters of size $\Omega(n/k)$ is another desirable property to avoid small, meaningless clusters. Motivated by this, we study the problem of minimizing the $k$-means objective, where the minimization is only over clusterings that satisfy $\alpha$-center proximity and are balanced.

\begin{definition}[$\alpha$-Center Proximity]\label{def:center-prox}
Let $C_{1}, C_{2}, \dotsc, C_{k}$ be a clustering of $X$ with the centers $\mu_{1}, \mu_{2}, \dotsc, \mu_{k}$ and the underlying metric $\text{dist}(\cdot, \cdot)$. We say that the clustering $C_{1}, C_{2}, \dotsc, C_{k}$ of $X$  satisfies $\alpha$-center proximity if for all $i \neq j$ and $x \in C_i$, we have $\text{dist}(x, \mu_{j}) > \alpha~ \text{dist}(x, \mu_{i})$.
We say that a $k$-clustering is {\em $\alpha$-center proximal} if it satisfies the $\alpha$-center proximity property.
\end{definition}

Unlike previous work, we do not assume that an optimal solution for the $k$-means objective on the given input satisfies $\alpha$-center proximity. There is no easy way to algorithmically verify this promise. In fact, as we show later in this paper, there are instances where the optimal $k$-means solutions satisfy $\alpha$-center proximity, for a small constant $\alpha > 1$, but the $k$-means problem still remains NP-hard. Therefore, we define our problem as finding a clustering of the smallest $k$-means cost among all the clusterings that satisfy $\alpha$-center proximity.

Given any set of cluster centers, the outliers are the last few points when we order all the points in non-decreasing order of their distances to the nearest centers, respectively. For the set of outliers to be unambiguous and stable under small perturbations to the input data, intuitively one needs a multiplicative gap between the distances of inliers and outliers to their respective centers. So we define an analogous center proximity property for clustering with outliers as follows.

\begin{definition}[$\alpha$-Center Proximity with Outliers]\label{def:alpha_ocp}
Consider a $k$-means instance on the set of points $X$ with underlying metric $\text{dist}(\cdot, \cdot)$, and an integer parameter $z$. We define the distance of a point $x$ to a set or a tuple $(\mu_1,\ldots,\mu_n)$ as $\min_{i} \norm{x - \mu_i}$.
Given any centers $\mu_{1}, \mu_{2}, \dotsc, \mu_{k}$, let $Z \subseteq X$ be the subset of the farthest $z$ points in $X$ based on their distances to the nearest center, and let $C_{1}, C_{2}, \dotsc, C_{k}$ be the clustering of $X \setminus Z$ where $C_{j}$ consists of the points in $X \setminus Z$ that have $\mu_{j}$ as their nearest center.
Such a clustering $C_{1}, C_{2}, \dotsc, C_{k}$ of $X \setminus Z$ with $Z$ as outliers satisfies $\alpha$-center proximity, if for all $i \neq j$ and $x \in C_i$, we have $\text{dist}(x, \mu_{j}) > \alpha~ \text{dist}(x, \mu_{i})$, and moreover, for all $i, j \in [k]$, $x \in C_{i}$ and $y \in Z$ we have $\text{dist}(y, \mu_{j}) > \alpha~ \text{dist}(x, \mu_{i})$.
\end{definition}

\subsection{Our results}

We show that there exists a constant $c > 1$ such that $k$-means remains NP-hard to approximate within a multiplicative factor $c$, even on the instances where an optimal $k$-means solution satisfies $\alpha$-center proximity and is balanced, i.e., each optimal cluster has size $\Omega(n/k)$. Moreover, for $\alpha$ close to $1$, there may not be a unique optimal $k$-means solution that satisfies $\alpha$-center proximity and is balanced. In fact, given any $\alpha > 1$, there exists an instance with $2^{\Omega\left(k/(\alpha - 1)\right)}$ such optimal $k$-means solutions that satisfy $\alpha$-center proximity and are balanced.

For any $\alpha > 1$, we show an interesting geometric property (\prettyref{prop:geom}) for clusterings that satisfy $\alpha$-center proximity, namely, any pair of disjoint clusters must lie inside two disjoint balls. The centers of these balls need not be at the means of the clusters, allowing the clusters to be arbitrarily large (see \prettyref{fig:amps}). The degenerate case for $\alpha = 1$ is two balls of infinite radii touching at their separating hyperplane.

We show the following algorithmic result for minimizing the $k$-means objective over the clusterings that are balanced and satisfy $\alpha$-center proximity.
\begin{theorem}[]
\label{thm:mpr}
For any $\alpha > 1$, a balance parameter $\omega > 0$, and given any set of $n$ points in $\R^{d}$ , we can \emph{exactly} find a clustering of the least $k$-means cost among all solutions that satisfy $\alpha$-center proximity and are balanced, i.e., each cluster has size at least $\omega n/k$. Our algorithm finds such an optimal clustering in time $O(2^{\text{poly}(k/\omega(\alpha - 1))}~ nd)$, with constant probability.
\end{theorem}

\begin{remark}
  We remark that our algorithm requires $\alpha$ as an input. However, in practice, the value of $\alpha$ might not be available in general. For an input $\alpha$, our algorithm can also be used to check whether an instance has an $\alpha$-center proximal clustering. On invoking our algorithm with a certain value of $\alpha$, if the instance has an $\alpha$-center proximal clustering, then our algorithm will output the optimal $\alpha$-center proximal clustering with constant probability. Therefore, the user can invoke our algorithm with sequence of decreasing values of $\alpha$ till a ``satisfactory" clustering is found.
\end{remark}

Since $k$-means is hard to approximate within some fixed constant $c > 1$, even on instances where the optimal solutions are balanced and satisfy $\alpha$-center proximity, the exponential running time in our algorithm is unavoidable. We show the following hardness result:

\begin{theorem}\label{thm:hard}
For any $2 >\alpha' >1$ there exists an $\alpha \leq \alpha'$, $(\alpha>1)$, constants $\e > 0$, and $\wmin > 0$, such that it is NP-hard to approximate the optimal $\alpha$-center proximal Euclidean $k$-means, where the size of each cluster is at least $\omega n / k$, to a factor better than $(1+\e)$.
\end{theorem}

The running time of our algorithm is exponential only in the number of clusters $k$, the balance parameter $\omega$ and the center proximity parameter $\alpha$ but it is linear in the number of points $n$ and the dimension $d$.

We show a similar exact algorithm for minimizing the $k$-means objective with $z$ outliers, where the minimization is only over clusterings that satisfy the $\alpha$-center proximity with outliers and are balanced.
\begin{theorem}[]
\label{thm:mpr-out}
For any $2> \alpha > 1$, a balance parameter $\omega > 0$, given any set of $n$ points in $\R^{d}$ and an outlier parameter $z \in [n]$, we can \emph{exactly} find a clustering of the least $k$-means cost among all solutions that satisfy $\alpha$-center proximity with $z$ outliers and are balanced, i.e., each cluster has size at least $\omega n/k$. Our algorithm finds such an optimal clustering in time $O(2^{\text{poly}(k/\omega(\alpha - 1))}~ nd)$, with constant probability.
\end{theorem}

In the case when most points satisfy $\alpha$-center proximity and form balanced clusters, we show an algorithm which outputs a list of clusterings, such that one of the clusterings corresponds to the case where the points which satisfy $\alpha$-perturbation resilience are correctly clustered.
\begin{theorem}\label{thm:beta-mpr}
For any $\alpha > 1$, a balance parameter $\omega > 0$, given any set of $n$ points in $\R^{d}$ and a parameter $\beta \in [n]$, we can output a list, of $k$-means clustering, of size $\bigo{2^{\text{poly}(k/\omega(\alpha - 1))}}$ such that one of them is the minimum cost clustering among all solutions that satisfy $\alpha$-center proximity without $\beta$ points and are balanced, i.e., each cluster has size at least $\omega n/k$. Our algorithm finds such an optimal clustering in time $O(2^{\text{poly}(k/\omega(\alpha - 1))}~ nd)$, with constant probability.
\end{theorem}

In fact, \prettyref{thm:mpr} and \prettyref{thm:mpr-out} hold for any clustering objective as long as the centers used to define $\alpha$-center proximity are the \emph{means} or \emph{centroids} of the clusters.

We also show exact algorithm for minimizing the $k$-means objective over clustering that satisfy $\alpha$-center proximity but no balance requirement. However, the running time of our algorithm depends exponentially on the ratio of the distances between the farthest and the closest pair of means.
\[
\gamma^* = \frac{\max_{i,j} \|\mu_{i} - \mu_{j}\|}{\min_{i \neq j} \|\mu_{i} - \mu_{j}\|}.
\]

\begin{theorem}[]
\label{thm:mpr-gamma}
For any $\alpha > 1$ and given any set of $n$ points in $\R^{d}$, and a parameter $\gamma$, where
\[
\gamma \geq \frac{\max_{i,j} \|\mu_{i} - \mu_{j}\|}{\min_{i \neq j} \|\mu_{i} - \mu_{j}\|} \mcom
\]
for the means $\mu_{1}, \mu_{2}, \dotsc, \mu_{k}$ of the optimal solution,
we can \emph{exactly} find a clustering of the least $k$-means cost among all solutions that satisfy $\alpha$-center proximity. Our algorithm finds such an optimal clustering in time $O(2^{\text{poly}(k \gamma/(\alpha - 1))}~ nd)$, with constant probability.
\end{theorem}

We also show that the optimal $\alpha$-center proximal clustering with balanced cluster need not be unique. In fact, the number of possible optimal $\alpha$-center proximal clustering with balanced cluster can be exponential in $k$ and $1/(\alpha-1)$. We show the following result:

\begin{proposition}\label{prop:size-alpha-opt}
  For any $2 >\alpha' >1$, and any $k \in \Z$, there exists $\alpha \leq \alpha'$ (and $\alpha > 1$), $n, d$ and a set of points $X \in \R^d$ such that  such that the number of possible optimal $\alpha$-center proximal clusterings, where the size of each cluster is $n/k$ ($\omega = 1$)  is $ 2^{ \tilde{\Omega}\paren{k / \paren{ \alpha^2-1}}} $.
\end{proposition}

\subsection{Related Work}

\paragraph{Metric Perturbation Resilience}
Bilu and Linial \cite{BiluL12} introduced the notion of multiplicative perturbation resilience for discrete optimization problems. They showed results on Max-Cut problems, that if the instance is roughly $\bigo{n}$ stable then they can retrieve the optimal Max-Cut in polynomial time. Later Makarychev et al. \cite{MakarychevMV14} gave an algorithm which required only $(c \sqrt{\log n} \log \log n)$- multiplicative perturbation resilience for some constant $c>0$ to retrieve the optimal Max-Cut in polynomial time.
Bilu and Linial \cite{BiluL12} had conjectured that there is some constant $\gamma^*$ such that, $\gamma^*$-perturbation resilient instances can be solved in polynomial time. They had also asked the question if it can be proved that Max-Cut is NP-hard even for $\gamma$-perturbation resilient instances, for some constant $\gamma$.
Awasthi et al. \cite{AwasthiBS12} showed that their conjecture in the case of ``well-behaved" center based objective functions like $k$-means, and $k$-median is true, by giving a polynomial time algorithm for $3$-perturbation resilient instances. They proposed the definition of center-based clustering objective and $\alpha$- center proximity (weaker notion than $\alpha$-multiplicative perturbation resilience). They showed that solving $k$-median instances with Steiner points over general metrics that satisfy $\alpha$-center proximity is NP-hard for $\alpha < 3$.
Ben-David and Reyzin \cite{Ben-DavidR14} showed that for every $\e > 0$, $k$-median instances with no Steiner points that satisfy $(2-\e)$-center proximity are NP-hard.
Balcan et al. \cite{BalcanHW16} showed that symmetric $k$-center under $(2-\e)$-multiplicative perturbation resilience is hard unless NP=RP. They also show an algorithm for solving symmetric and asymmetric $k$-center for $\alpha \geq 2$.
Balcan and Liang \cite{BalcanL16} improved upon the results of \cite{AwasthiBS12} and showed that center based objective can be optimally clustered for $\alpha \geq 1+\sqrt{2}$ factor perturbation resilience. They showed in their Lemma 3.3 that for a pair of cluster optimal clusters $C_i, C_j$, and $\alpha \geq 1+\sqrt{2}$, the clusters are contained in disjoint balls around their centers $\mu_i$ and $\mu_j$ respectively. We show in \prettyref{prop:amps} that for any $\alpha > 1$, and for a pair of optimal clusters $C_i, C_j$, the points in the clusters are contained in disjoint balls, centered around a point different from the mean of the clusters.
We note that a similar geometric property was observed by Telgarsky and Vattani \cite{Telgarsky10}. Telgarsky and Vattani \cite{Telgarsky10} defined a version of Harting's method, which does updates based on the value $\alpha$. The value of $\alpha$ in their case is defined in terms of the size of the clusters. Although in a different context, the geometric insight they achieved is the same as in \prettyref{prop:amps}.
Recently Angelidakis et al. \cite{Angelidakis17} gave a more general definition of center based objective functions and of metric  perturbation resilience compared to the one given by \cite{AwasthiBS12}. They improved the previous results \cite{BalcanL16}, giving an algorithm for center based clustering under $2$-multiplicative perturbation resilience.

Independent and concurrent to this work, Friggstad et al. \cite{Friggstad18} showed that for any fixed $d \geq 1$ and $\alpha >1$, $\alpha$-multiplicative perturbation resilient instances of discrete $k$-means and $k$-median (where the centers must be from the data-points themselves) in metrics with doubling dimension $d$ can be solved in time $\bigo{n^{d^{\bigo{d}}(\alpha-1)^{-\bigo{d/(\alpha-1)}}}k}$.
They also showed that when the dimension $d$ is a part of the input, there is a fixed $\e_0 > 0$ such there is not even a PTAS
for $(1 + \e_0)$-multiplicative perturbation resilient instances of $k$-means in $\R^d$ unless NP=RP. We note that our hardness result does not subsume their hardness of approximation result as multiplicative perturbation resilient instances are a subset of center-proximal instances (that is, the optimal $k$-means solution satisfies center-proximity \cite{Angelidakis17}). Our problem is different from theirs; we look for center proximal instances which is a more general class of instances than perturbation resilient instances \cite{Angelidakis17}. Moreover, we do not assume that the optimal solution $k$-means solution is center proximal, unlike Friggstad et al. \cite{Friggstad18}, where they assume that the optimal $k$-means instance is perturbation resilient.
Note that it is not known (to the best of our knowledge) whether there is an efficient algorithm to check whether an instance satisfies $\alpha$-metric perturbation resilience.

\paragraph{Approximation Algorithms for Euclidean $k$-means} Kanungo et al. \cite{Kanungo04} proposed a $(9+\e)$-approximation algorithm with running time $\bigo{n^3 \e^{-d}}$. Arthur and Vassilvitskii \cite{ArthurV04} showed an approximation ratio of $\bigo{\log k}$, with running time $\bigo{ndk}$, much superior to \cite{Kanungo04}.
Recently Ahmadian et al. \cite{Ahmadian17} improved the approximation ratio given by \cite{Kanungo04} to $(6.357+\e)$. There have been works to get a PTAS for Euclidean $k$-means objective.
In order to obtain the PTAS, many have focused on cases where $k$ or $d$ or both are assumed to be fixed. Inaba et al. \cite{Inaba94} gave a PTAS when both $k$ and $d$ is fixed. There have been a series of work in the case when only $k$ is assumed to be fixed \cite{Vega03, Peled04,Peled05,FeldmanMS07, KumarSS04,Chen09}.
Recently there has been works which give a PTAS for Euclidean $k$-means where only $d$ is assumed to be a constant \cite{Addad16,Friggstad16}.

\paragraph{Sampling Based Methods} As mentioned before that our results use sampling techniques often used in $(1+\e)$- approximation for $k$-means \cite{KumarSS04, Chen09, AckermannBS10}. The first ever linear (in $n$ and $d$) running time for obtaining PTAS (assuming $k$ to be a constant) given by \cite{KumarSS04} is $\bigo{nd 2^{poly(k/ \e)}} $. Feldman et al. \cite{FeldmanMS07} gave a new algorithm (using efficient coreset construction) with a better running time than that of \cite{KumarSS04} from $\bigo{nd 2^{poly(k/ \e)}} $ to $\bigo{nkd + d.poly(k/ \e) + 2^{\tbigo{k/\e}}}$.
There have been other works which also show similar results using $D^2$ sampling method \cite{JaiswalKS14,BhattacharyaJK18}. Ding and Xu \cite{Ding15} gave a sampling based procedure to cluster other variants of $k$-means objective, which they called the constrained $k$-means clustering. These clustering objectives need not satisfy the locality property in Euclidean space. Their algorithm is based on uniform sampling and some stand alone geometric technique which they call the `simplex lemma'.
The running time of their algorithm is $\bigo{2^{\poly(k/\e)} n (\log n)^{k+1}d}$. Bhattacharya et al. \cite{BhattacharyaJK18} later gave a more efficient algorithm based on $D^2$ sampling for the same class of constrained $k$-means problem and gave an algorithm with running time $\bigo{2^{\tbigo{k/\e}} nd}$. Deshpande et al.\cite{me} also use the $D^2$ sampling method to solve the min-max $k$-means problem, wherein, one has to find a clustering such that the maximum cost of the cluster is minimized. All the work related to sampling based methods above estimate the means/centers of the clusters, and use them to recover a clustering.

\paragraph{Clustering with Outliers}
For the problem of $k$-median with outliers, Charikar et al. \cite{Charikar01} gave an algorithm that removes $(1+\e)$ times the number of outliers, and has a cost at most $4(1+1/\e)$ times the optimal cost. Later, Chen \cite{Chen08} gave a $\bigo{1}$ approximation for the $k$-median with outliers problem. However, the unspecified approximation factor of their algorithm is large. For the problem of $k$-means with outliers, Gupta et al. \cite{Gupta17} gave a $\bigo{1}$ approximation by removing $\bigo{z~k \log n}$ number of outliers (where $z$ is the number of outliers).
Friggstad et al. \cite{FriggstadKRS18} gave an algorithms that open $k(1 + \e)$ facilities, and have
an approximation factor of $(1 + \e)$ on Euclidean and doubling metrics, and $25 + \e$ on general metrics. Recently, Krishnaswamy et al. \cite{Krishnaswamy18} gave a $(53.002 + \e)$ factor approximation for the $k$-means with outliers problem, and a $(7.081 + \e)$ approximation factor for the $k$-median with outliers problem, without violating any constraint.
Chekuri et al. \cite{Chekuri18} defined a new model of perturbation resilience for clustering with outliers, and exactly solve the clustering with outliers problem for several common center based objectives (like $k$-center, $k$-means, $k$-median) when the instances is $2$- multiplicative perturbation resilient. We note that our definition of outliers is inspired by Chekuri et al. \cite{Chekuri18}.

\paragraph{Other Notions of Stability} Ackerman and Ben-David \cite{AckermanB-D2009} studied various deterministic assumptions to obtain solutions with small objective costs, one of them being additive perturbation resilience. We refer the reader to \cite{AckermanB-D2009} for a more detailed survey of various notions of stability and their implications. Vijayraghavan et al. \cite{Dutta17} gave a more general definition of $\e$-additive perturbation resilience, where the points in the instance even after being moved by at
most $\e \cdot \max_{i,j} \norm{\mu_i-\mu_j}$ (where $\e \in (0,1)$ is a parameter) reamin in the same optimal clusters. Vijayaraghavan et al. \cite{Dutta17} showed a geometric property of $\e$-additive perturbation resilient instances implies an \emph{angular separation} between points from any pair of clusters. Using this observation, they showed that a modification of the perceptron algorithm (from supervised learning) can optimally solve any $\e$-additive perturbation resilient instance of $2$-means in time $dn^{\text{poly}(1/\e)}$. We make an observation (\prettyref{prop:geom}) that $\alpha$-metric perturbation resilient instances of $k$-means in the Euclidean space also satisfy \emph{angular separation} similar to \cite{Dutta17}. For $k$-means their running time is $dn^{k^{2}/\e^{8}}$.
To get a faster algorithm, they defined a stronger notion of stability than $\e$-additive perturbation resilience called $(\rho, \Delta, \e)$-separation, a natural strengthening of additive perturbation stability where there is an additional margin of $\rho$ between any pair of clusters.
They give an algorithm based on the $k$-largest connected components in a graph that can optimal solve any $(\rho, \Delta, \e)$-separated instance of $k$-means with $\beta$-balanced clusters in time $\tilde{O}(n^{2} kd)$ whenever $\rho = \Omega(\Delta/\e^{2} + \beta \Delta/\e)$. They also showed that their algorithm is robust to outliers as long as the fraction $\eta$ of outliers satisfies the following equation
\[
\rho = \Omega\left(\frac{\Delta}{\e^{2}} \left(\frac{w_{\max} + \eta}{w_{\min} - \eta}\right)\right) \mcom
\]
where $w_{\max}$ and $w_{\min}$ are the fraction of points in the largest and the smallest optimal cluster, respectively. Balcan et al. \cite{BalcanBG13} explored the $(c,\e)$-approximation stability which assumes that every $c$-approximation to the cost is $\e$-close (in normalized set difference) to the target clustering . Balcan and Liang \cite{BalcanL16} showed that when the target clustering is the optimal clustering then the $(c,\e)$-approximation stability implies $(c,\e)$-perturbation resilience, that is, the optimum after perturbation of up to a multiplicative factor $c$ is $\e$-close (in normalized set difference) to the original clustering. Kumar and Kannan \cite{KumarK10} also gave deterministic conditions (\textit{read stability})
under which their algorithm finds the optimal clusters. The analysis of Kumar and Kannan \cite{KumarK10} was tightened by \cite{AwasthiS12} which basically needed that the cluster centers must be pair-wise separated by a margin of $\Omega(\sqrt{k} \sigma)$ along the line joining the mean of the clusters, where $\sigma$ denotes the ``spectral radius" of the data-set.

\subsection{Notation}
We use $d \in \Z$ to denote the dimension of the ambient space.
For a vector $v \in \R^d$, we use $\norm{v}$ to denote its Euclidean norm.
We use $x_1, \ldots, x_n \in \R^d$ to denote the points in the instance,
and we use $X$ to denote the set of these points.
For any set of points $S$, we define its {\em diameter} as
$\diam{S} \defeq \max_{u,v \in S} \norm{u - v}$.

\section{Geometric Properties of $\alpha$-Center Proximal Instances}
We will assume Euclidean metric throughout the paper.

\begin{definition}\label{def:hatmu}
Suppose there exists a clustering $C_1,...,C_k$ of $X$ that is $\alpha$-center proximal. Let $C_i$ and $C_j$ be two clusters with $\mu_i$ and $\mu_j$ as their respective means. We define the vector $\Hat\mu_{i,j} \defeq \frac{\alpha^2 \mu_i- \mu_j}{\alpha^2-1}$ and $r_{i,j} \defeq \frac{\alpha}{\alpha^2-1} \norm{\mu_i-\mu_j}$.
Let $C_{i,j}$ denote the ball centered at $\Hat\mu_{i,j}$ with radius $r_{i,j}$. Let $p_{i,j} \defeq (\Hat\mu_{i,j}+\Hat\mu_{j,i})/2$ and $D_{i,j} \defeq \norm{\Hat\mu_{i,j}-\Hat\mu_{j,i}}$.
Let $u \defeq \frac{(\Hat\mu_{i,j}-\Hat\mu_{j,i})}{\norm{\Hat\mu_{i,j}-\Hat\mu_{j,i}}}$ denote the unit vector along the line joining $\Hat\mu_{i,j}$ and $\Hat\mu_{j,i}$, and let $d_{i,j}$ denote the distance between the closest points in $C_{i,j}$ and $C_{j,i}$.
\end{definition}

We note that Balcan and Liang \cite{BalcanL16} show in their Lemma 3.3 that for a pair of clusters $C_i, C_j$ in a $\alpha$-center proximal clustering, and $\alpha \geq 1+\sqrt{2}$, the clusters are contained in disjoint balls around their centers $\mu_i$ and $\mu_j$ respectively. We show in the following proposition that for any $\alpha > 1$, and for a pair of clusters $C_i, C_j$ in a $\alpha$-center proximal clustering, the points in the clusters are contained in disjoint balls, centered around $\Hat\mu_{i,j}$ and $\Hat\mu_{j,i}$.

\begin{proposition}[Geometric implication of $\alpha$-center proximity property]\label{prop:amps}
Let $X$ satisfy the $\alpha$-center proximity property (for any $\alpha > 1$) and let $C_i$ and $C_j$ be two clusters in its optimal solution. Any point $x \in C_i$, satisfies
  \begin{equation}\label{amps}
  \forall x \in C_i, \norm{ x - \Hat\mu_{i,j}} < r_{i,j} \mper
  \end{equation}
\end{proposition}

\begin{proof}
The proof proceeds simply by using the $\alpha$-center proximity property for Euclidean metric and squaring both the sides, we get
\[ \alpha^2 \norm{x-\mu_i}^2 < \norm{x -\mu_j}^2 \mcom \]

\[ \alpha^2\norm{x}^2 + \alpha^2\norm{\mu_i}^2 -2\alpha^2\inprod{x,\mu_i} < \norm{x}^2 + \norm{\mu_j}^2 - 2\inprod{x,\mu_j} \mcom \]

\[ \norm{x}^2  - 2\inprod{x,\frac{\alpha^2\mu_i -\mu_j}{(\alpha^2-1)}} < \frac{\norm{\mu_j}^2 - \alpha^2\norm{\mu_i}^2}{(\alpha^2-1)}  \mper \]

Completing the square on LHS and adding the appropriate term on RHS
\[ \norm{x}^2  - 2\inprod{x,\frac{\alpha^2\mu_i -\mu_j}{(\alpha^2-1)}} + \norm{\frac{\alpha^2\mu_i -\mu_j}{(\alpha^2-1)}}^2 < \frac{\norm{\mu_j}^2 - \alpha^2\norm{\mu_i}^2}{(\alpha^2-1)} + \norm{ \frac{\alpha^2\mu_i -\mu_j}{(\alpha^2-1)}}^2 \mcom \]

\[ \norm{x- \frac{\alpha^2\mu_i -\mu_j}{(\alpha^2-1)}}^2 < \frac{\alpha^2\norm{\mu_i-\mu_j}^2}{(\alpha^2-1)^2} \mper \]

Therefore using the terms in \prettyref{def:hatmu} we get the proposition.

\end{proof}

\definecolor{uuuuuu}{rgb}{0.26666666666666666,0.26666666666666666,0.26666666666666666}
\definecolor{xdxdff}{rgb}{0.49019607843137253,0.49019607843137253,1}
\begin{figure*}[!htb]
\centering
\begin{tikzpicture}[line cap=round,line join=round,>=triangle 45,x=1cm,y=1cm]

\draw  (-4,0) circle (2cm);
\draw  (4,0) circle (2cm);
\draw  (-4,0)-- (4,0);
\draw  (0,0)-- (-2.963591896316167,-1.7105140287698555);
\draw [->] (-4,0) -- (-2.963591896316167,-1.7105140287698555);
\draw [shift={(0,0)}]  plot[domain=3.141592653589793:3.665060695266103,variable=\t]({1*0.8054424827460394*cos(\t r)+0*0.8054424827460394*sin(\t r)},{0*0.8054424827460394*cos(\t r)+1*0.8054424827460394*sin(\t r)});
\draw [dashed] (-4,0) -- (-4,2);
\draw [dashed] (4,0) -- (4,2);
\draw [dashed] (-4,2) -- (-2,2);
\draw [dashed] (4,2) -- (2,2);
\draw (-1.4910148482484508,2.282279926033007) node[anchor=north west] {$D_{i,j} = \left(\frac{\alpha^2+1}{\alpha^2-1}\right) \norm{\mu_i - \mu_j}$};
\draw (-4.653679006500432,0.25507812831737947) node[anchor=north west] {$\hat \mu_{i,j}$};
\draw (-2.947291924891407,0.06972577274710827) node[anchor=north west] {$\mu_i$};
\draw (4.157881705302359,0.25507812831737947) node[anchor=north west] {$\hat \mu_{j,i}$};
\draw (2.6132787422167585,0.06972577274710827) node[anchor=north west] {$\mu_j$};
\draw[decoration={brace,raise=5pt},decorate]
  (-4,0) -- node[above=6pt] {$\frac{D_{i,j}}{\alpha^2+1}$} (-2.8,0);
\draw (-5.389141367036774,-0.46573658778923077) node[anchor=north west] {$ \left( \frac{\alpha D_{i,j}}{\alpha^2+1} \right) = r_{i,j}$};
\draw (-1.361499549456856,-0.09032047892158285) node[anchor=north west] {$\theta$};
\draw (-1.1967419000610584,-0.6245471754402816) node[anchor=north west] {$\theta = \tan^{-1}\left({\frac{2 \alpha}{\alpha^2-1}}\right)$};
\draw [->] (0,0) -- (-1.3,0);
\draw (-1,0) node[above=2pt] {$u$};
\draw [dashed] (-2,0) -- (-2,-2);
\draw [dashed] (2,0) -- (2,-2);
\draw [dashed] (-2,-2) -- (-1.108416013618413,-2.010339550874824);
\draw [dashed] (2,-2) -- (1.1245128359098342,-2.010339550874824);
\draw [red,dashed] (0,0) -- (-5.980269880825548,3.452710425521241);
\draw [red,dashed] (0,0) -- (-5.980269880825548,-3.452710425521241);
\draw [red,dashed] (0,0) -- (5.980269880825548,3.452710425521241);
\draw [red,dashed] (0,0) -- (5.980269880825548,-3.452710425521241);
\draw (-1.1642530703970901,-1.5924983656405867) node[anchor=north west] {$d_{i,j}=\frac{(\alpha-1)^2}{\alpha^2+1} D_{i,j}$};
\draw (-0.8672266012694635,0.9917299012026677) node[anchor=north west] {$p_{i,j} = \frac{\mu_i + \mu_j}{2}$};
\draw (-5.562819609049692,0.9758928444239897) node[anchor=north west] {$C_{i,j}$};
\draw (4.967022307851619,0.9758928444239897) node[anchor=north west] {$C_{j,i}$};
\begin{scriptsize}
\draw [fill=xdxdff] (-4,0) circle (1pt);
\draw [fill=xdxdff] (4,0) circle (1pt);
\draw [fill=xdxdff] (-2.803128981670084,0) circle (1pt);
\draw [fill=xdxdff] (2.819225803961505,0) circle (1pt);
\draw [fill=xdxdff] (0,0) circle (1pt);
\draw [fill=uuuuuu] (-2,0) circle (1pt);
\draw [fill=uuuuuu] (2,0) circle (1pt);
\end{scriptsize}
\end{tikzpicture}
\caption{Geometric implication of $\alpha$-center proximity property.}
\label{fig:amps}
\end{figure*}
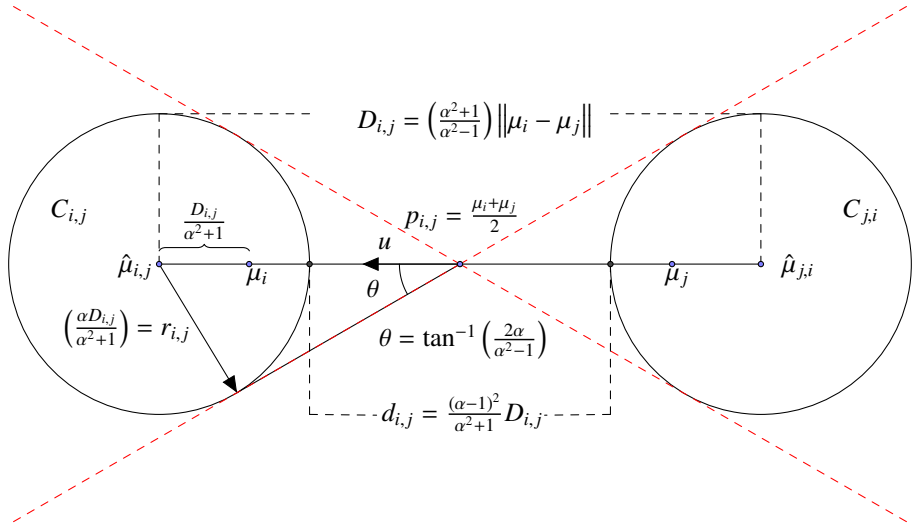

\begin{proposition}\label{prop:geom}
Suppose there exists a clustering $C_1,...,C_k$ of $X$ that is $\alpha$-center proximal (for any $\alpha > 1$). From \prettyref{prop:amps} we know that for any two clusters $C_i$ and $C_j$, the clusters $(i \neq j)$ lie inside disjoint balls centered at $\Hat\mu_{i,j}$ and $\Hat\mu_{j,i}$ with radius $r_{ij}$ and $r_{ji}$, respectively. The following structural properties for $\alpha$-center proximal clustering holds:
\begin{enumerate}
  \item[(a)] Distance between the respective centers of the balls is $D_{i,j} = \norm{\Hat \mu_{i,j} - \Hat\mu_{j,i}} = \frac{\alpha^2+1}{\alpha^2-1}\norm{\mu_i-\mu_j}$.
  \item[(b)] The mid-point of the line joining the respective centers of the two balls is $p_{i,j}= \frac{\Hat \mu_{i,j} + \Hat\mu_{j,i}}{2} = \frac{\mu_i+\mu_j}{2}$.
  \item[(c)] The distance between the closest points in the two balls is greater than $d_{i,j} = \frac{\alpha-1}{\alpha+1}\norm{\mu_i-\mu_j} =\frac{(\alpha-1)^2}{\alpha^2+1} D_{i,j}$.
  \item[(d)]The radius of the ball $C_{i,j}$ is $r_{i,j}=\frac{\alpha}{\alpha^2-1}\norm{\mu_i-\mu_j}= \frac{\alpha}{(\alpha-1)^2}d_{i,j} = \frac{\alpha }{\alpha^2+1}D_{i,j}$.
  \item[(e)] Distance between the mean of the cluster and the center of the ball corresponding to the cluster is $\norm{\Hat \mu_{i,j} - \mu_i}=\frac{1}{\alpha^2-1}\norm{\mu_i-\mu_j}=\frac{1}{(\alpha-1)^2}d_{i,j} = \frac{1}{\alpha^2+1}D_{i,j}$.
  \item[(f)] The two clusters lie inside a cone with apex at $p_{i,j}$ whose axis is along the line joining $(\Hat\mu_{i,j}-\Hat\mu_{j,i})$ with half-angle $\tan^{-1}\left(\frac{2\alpha}{\alpha^2-1}\right)$.
  \item[(g)] The diameter of the cluster $C_i$ can be bounded by the radius of the ball $C_{i,j}$, $\diam{C_i} \leq \frac{2\alpha}{\alpha^2-1}\norm{\mu_i-\mu_j} = \frac{2\alpha }{\alpha^2+1}D_{i,j}$.
  \item[(h)] For any $x \in C_i, \norm{x-\mu_j} > \frac{\alpha}{\alpha+1} \norm{\mu_i-\mu_j}$.
\end{enumerate}

\end{proposition}
\begin{proof}
\begin{enumerate}
  \item[(a)] $D_{i,j} = \norm{\Hat\mu_{i,j} - \Hat\mu_{j,i}}$.\
  \[ D_{i,j} = \norm{\frac{\alpha^2 \mu_i- \mu_j}{\alpha^2-1} - \frac{\alpha^2 \mu_j- \mu_i}{\alpha^2-1}} = \left( \frac{\alpha^2+1}{\alpha^2-1} \right) \norm{\mu_i-\mu_j} \mper \]
  \item[(b)]
  \[ p_{i,j} = \frac{\Hat\mu_{i,j} + \Hat\mu_{j,i}}{2} = \frac{1}{2}\left(\frac{\alpha^2 \mu_i- \mu_j}{\alpha^2-1} + \frac{\alpha^2 \mu_j- \mu_i}{\alpha^2-1}\right) = \frac{\mu_i+\mu_j}{2} \mper \]
   Note that $p_{i,j}$ is also the mid-point of the line joining the means of the two clusters.
  \item[(c)]
  The distance between the closest points in the two balls is greater than $d_{i,j} = D_{i,j}-2r_{i,j}$. Putting in the value of $D_{i,j}$ from part (a) of the proposition, and $r_{i,j}$ from the previous proposition, we get that
  \[ d_{i,j} > \paren{ \frac{\alpha^2+1}{\alpha^2-1}   - \frac{2\alpha}{\alpha^2-1}}\norm{\mu_i-\mu_j}  = \frac{(\alpha-1)^2}{\alpha^2-1}\norm{\mu_i-\mu_j} = \frac{(\alpha-1)^2}{\alpha^2+1} D_{i,j} \mper \]
  \item[(d)] The proof follows from the definition of $r_{i,j}$ and part (c) of the proposition.
  \item[(e)] \[ \norm{\Hat \mu_{i,j} - \mu_i} = \norm{\frac{\alpha^2 \mu_i-\mu_j}{\alpha^2-1} - \mu_i}\ = \frac{1}{\alpha^2-1}\norm{\mu_i-\mu_j} \mper \]
  The proof follows using the part (c) of the proposition.

  \item[(f)] Let the apex of the cone (with its axis along $u$) be at $p_{i,j}$. Draw the tangent from $p_{i,j}$ to $C_{i,j}$. Using the fact that the tangent is perpendicular to radius drawn on the point of contact of the ball and the tangent line, we get for the half-angle $\theta$ that $\sin{\theta} = \frac{r_{i,j}}{D_{i,j}/2} = \frac{\alpha / (\alpha^2+1)}{1 / 2}$. Therefore $\tan{\theta} = \frac{2\alpha}{\alpha^2-1}$.
  \item[(g)] The equation \ref{amps} implies that all the points belonging to the cluster $C_{i}$ must lie inside the ball $C_{i,j}$ of radius $r_{i,j}$. Therefore the proof follows using part (d) of the proposition.
  \item[(h)] The smallest distance between any point in $C_i$ and $\mu_j$ is greater than $r_{i,j} + d_{i,j} - \norm{\Hat \mu_{j,i} - \mu_j} $. Using part (d) and (e) of the proposition, we get
  \[ r_{i,j} + d_{i,j} - \norm{\Hat \mu_{j,i} - \mu_j} = \paren{\frac{\alpha}{\alpha^2-1} + \frac{(\alpha-1)^2}{\alpha^2-1} - \frac{1}{\alpha^2-1}} \norm{\mu_i-\mu_j} = \frac{\alpha}{\alpha+1} \norm{\mu_i-\mu_j} \mper \]
\end{enumerate}
\end{proof}

\begin{remark}\label{rmk:amps}
The line joining the respective centers of the ball passes through the mean of the clusters as well. Refer to the \prettyref{fig:amps} for getting an insight into the geometric structure defined by the equation \ref{amps}. We refer the reader to figure 1a of \cite{Dutta17} for insight into the geometry of $\e$-additive perturbation resilient instance, and to notice the similarity between the geometric structures of instances satisfying the two stability properties.
\end{remark}

\subsection{Estimating Means Suffices for Cluster Recovery}

\RestyleAlgo{boxruled}
\begin{algorithm}
\caption{$\alpha$-center proximal $k$-means clustering with balanced clusters}
\label{alg:balance-means-known}
\begin{algorithmic}[1]
\REQUIRE a list $\cL$ of $k$-tuples, where there exists at least one tuple $\paren{\tmu_1,\ldots,\tmu_k}$ such that for any $i \in \brac{k}$, \\  $\norm{\tmu_i - \mu_i} \leq 2 \delta \rij $ for all $j \in \brac{k}$, a number $\alpha >1$, $\omega > 0$
\ENSURE An $\alpha$-center proximal clustering of minimum cost where the size of each cluster is at least $\omega n / k$.
\STATE Initialize $(\tilde{\mu}_{1}^{(p)}, \tilde{\mu}_{2}^{(p)}, \dotsc, \tilde{\mu}_{k}^{(p)}) \leftarrow (\bar{0}, \bar{0}, \dotsc, \bar{0})$
\STATE $\text{cost} \leftarrow \sum_{i=1}^{n} \min_{1 \leq j \leq k} \norm{x_{i} - \tilde{\mu}_{j}^{(0)}}^{2}$
\FOR{ $\text{a tuple } \tilde \mu^{(q)} \in {{\cL}}$}
\FOR{$j=1 \text{ to } k$}
\STATE $C_{j}^{(q)} = \{x_{i} \suchthat \tilde{\mu}_{j}^{(q)}~ \text{is the nearest center to}~ x_{i}\}$
\ENDFOR
\IF{$\sum_{i=1}^{n} \min_{1 \leq j \leq k} \norm{x_{i} - \tilde{\mu}_{j}^{(q)}}^{2} < \text{cost}$}
\IF{$C_{1}^{(q)},\ldots,C_{k}^{(q)} \text{ all have size at least } \frac{\omega n}{k} \text{ and satisfy } \alpha\text{-center proximity}$}
\STATE $(\tilde{\mu}_{1}^{(p)}, \tilde{\mu}_{2}^{(p)}, \dotsc, \tilde{\mu}_{k}^{(p)}) \leftarrow (\tilde{\mu}_{1}^{(q)}, \tilde{\mu}_{2}^{(q)}, \dotsc, \tilde{\mu}_{k}^{(q)}) $
\STATE cost $= \sum_{i=1}^{n} \min_{1 \leq j \leq k} \norm{x_{i} - \tilde{\mu}_{j}^{(q)}}^{2}$
\ENDIF
\ENDIF
\ENDFOR
\FOR{$i=1$ to $n$}
\STATE $\text{Label}(i) \leftarrow \underset{1 \leq j \leq k}{\text{argmin}} \norm{x_{i} - \tilde{\mu}_{j}^{(p)}}$
\ENDFOR
\end{algorithmic}
\end{algorithm}

\begin{algorithm}
\caption{$\alpha$-center proximal $k$-means clustering where the distance between the means is bounded}
\label{alg:gamma-means-known}
\begin{algorithmic}[1]
\REQUIRE a list $\cL$ of $k$-tuples, where there exists at least one tuple $\paren{\tmu_1,\ldots,\tmu_k}$ such that for any $i \in \brac{k}$, \\  $\norm{\tmu_i - \mu_i} \leq 2 \delta \rij $ for all $j \in \brac{k}$, a number $\alpha >1$, $\gamma \geq 1$
\ENSURE An $\alpha$-center proximal clustering of minimum cost where $ \gamma \leq \dfrac{\max_{i,j} \norm{\mu_i-\mu_j}}{\min_{i \neq j} \norm{\mu_i - \mu_j}}$
\STATE Initialize $(\tilde{\mu}_{1}^{(p)}, \tilde{\mu}_{2}^{(p)}, \dotsc, \tilde{\mu}_{k}^{(p)}) \leftarrow (\bar{0}, \bar{0}, \dotsc, \bar{0})$
\STATE $\text{cost} \leftarrow \sum_{i=1}^{n} \min_{1 \leq j \leq k} \norm{x_{i} - \tilde{\mu}_{j}^{(0)}}^{2}$
\FOR{ $\text{a tuple } \tilde \mu^{(q)} \in {{\cL}}$}
\FOR{$j=1 \text{ to } k$}
\STATE $C_{j}^{(q)} = \{x_{i} \suchthat \tilde{\mu}_{j}^{(q)}~ \text{is the nearest center to}~ x_{i}\}$
\ENDFOR
\STATE Find the actual means $(\mu_1^{(q)},\ldots,\mu_k^{(q)})$ of $C_{1}^{(q)},\ldots,C_{k}^{(q)}$
\IF{$\sum_{i=1}^{n} \min_{1 \leq j \leq k} \norm{x_{i} - \tilde{\mu}_{j}^{(q)}}^{2} < \text{cost}$}
\IF{$C_{1}^{(q)},\ldots,C_{k}^{(q)} \text{ satisfy } \alpha\text{-center proximity and } \gamma \geq \dfrac{\max_{i,j} \norm{\mu_i^{(p)}-\mu_j^{(p)}}}{\min_{i \neq j} \norm{\mu_i^{(p)} - \mu_j^{(p)}}}$}
\STATE $(\tilde{\mu}_{1}^{(p)}, \tilde{\mu}_{2}^{(p)}, \dotsc, \tilde{\mu}_{k}^{(p)}) \leftarrow (\tilde{\mu}_{1}^{(q)}, \tilde{\mu}_{2}^{(q)}, \dotsc, \tilde{\mu}_{k}^{(q)}) $
\STATE cost $= \sum_{i=1}^{n} \min_{1 \leq j \leq k} \norm{x_{i} - \tilde{\mu}_{j}^{(q)}}^{2}$
\ENDIF
\ENDIF
\ENDFOR
\FOR{$i=1$ to $n$}
\STATE $\text{Label}(i) \leftarrow \underset{1 \leq j \leq k}{\text{argmin}} \norm{x_{i} - \tilde{\mu}_{j}^{(p)}}$
\ENDFOR
\end{algorithmic}
\end{algorithm}

In this section we assume that we can get access to a set of points $\set{\tmu_1,\ldots,\tmu_k}$ such that, $\norm{\tmu_i - \mu_i} \leq 2 \delta \rij $ for all $j \in \brac{k}$ and for each $i \in \brac{k}$, with probability some constant probability. We will show how to obtain such points in \prettyref{sec:main-sampling}.

The following proposition shows that to cluster a point to their desired cluster, an additive approximation to the center also works (deciding based on the proximity of the data points to the approximate center).
\begin{proposition}\label{prop:combined-sampling}
	Suppose we have a set of points $\set{\tmu_1,\ldots,\tmu_k}$ such that for any $i \in \brac{k}$,  $\norm{\tmu_i - \mu_i} \leq 2 \delta \rij $ for all $j \in \brac{k}$, then we can find the optimal clustering $C_1,\ldots, C_k$.
\end{proposition}
\begin{proof}
We will break the proof of the proposition into the following lemmas.
\begin{lemma}
\label{lem:cor-sampling}
For any $j \in [k]$, for all $i \in \brac{k}$, and for any point $x$, we have
	\[ \Abs{\norm{x - \tmu_j} - \norm{x - \mu_j}} \leq 2\delta \rij \mper   \]
\end{lemma}

\begin{proof}
For any two vectors $u, v \in \R^d$, the triangle inequality implies that $\Abs{ \norm{u} - \norm{v}} \leq \norm{u - v}$.
Therefore,
\[ \Abs{\norm{x - \tmu_j} - \norm{x - \mu_j}} \leq \norm{\tmu_j - \mu_j}
	\leq 2 \delta \rij \mper   \]
\end{proof}

\begin{lemma}
\label{lem:muprox}
Fix any $i \in [k]$. For any point $x \in C_i$, we have
\[ \norm{x - \tmu_i} < \norm{x - \tmu_j} \qquad \textrm{for each } j \in [k] \setminus \set{i} \mper   \]
\end{lemma}

\begin{proof}
By our choice of parameters in \prettyref{alg:stable-km}, we have
\begin{equation}
\label{eq:alphadelta}
	\delta \leq \frac{\paren{\alpha - 1}^2}{8 \alpha} \mper
\end{equation}

Fix any $j \in [k] \setminus \set{i}$. Using \prettyref{def:center-prox}, we get
\begin{equation}
\label{eq:alpha2}
	\norm{x - \mu_i} < \frac{1}{\alpha} \norm{x - \mu_j} \mper
\end{equation}

\begin{align*}
	\norm{x - \tmu_i} & \leq \norm{x - \mu_i} + 2 \delta \rij & \paren{\textrm{\prettyref{lem:cor-sampling}}} \\
	& < \frac{1}{\alpha} \norm{x - \mu_j} + 2 \delta \rij
	& \paren{\textrm{Using \prettyref{eq:alpha2}}}\\
	& = \paren{\norm{x - \mu_j} - 2 \delta \rij} + \paren{\frac{1}{\alpha} -1 } \norm{x - \mu_j} \\
	& \qquad \qquad	+ 2\delta \rij + 2\delta \rij \\
	& \leq \norm{x - \tmu_j} + 4 \delta \rij -\frac{\alpha - 1}{\alpha} \paren{\alpha - 1} \rij
	& \paren{\textrm{Using \prettyref{lem:cor-sampling} and \prettyref{prop:geom}}} \\
	& \leq  \norm{x - \tmu_j} + \rij \paren{4 \delta - \frac{\paren{\alpha - 1}^2}{\alpha}} \\
	& \leq \norm{x - \tmu_j} +  \rij \cdot \paren{ -\frac{\paren{\alpha - 1}^2}{2\alpha}}
		< \norm{x - \tmu_j} \mper
	& \paren{\textrm{Using \prettyref{eq:alphadelta}}}
\end{align*}
\end{proof}

The next lemma shows that the clustering obtained by approximate centers corresponds to the optimal clustering, ignoring the outliers.
\begin{lemma}
\label{lem:cip}
For any $i \in [k]$, $\tilde C_i \cap \paren{X \setminus Z} = C_i$.
\end{lemma}

\begin{proof}
\prettyref{lem:muprox} implies that each $x \in C_i$ is closest to $\tmu_i$
out of $\set{\tmu_1, \ldots, \tmu_k}$. Therefore, $x \in \tilde C_i$, and hence
\begin{equation}
\label{eq:muprox1}
	C_i \subseteq \tilde C_i \cap \paren{X \setminus Z} \qquad \textrm{for each } i \in [k] \mper
\end{equation}
By construction, each $x \in X$ belongs to exactly one out of $\set{\tilde C_1, \ldots, \tilde C_k}$.
Therefore,
\begin{equation}
\label{eq:muprox2}
\paren{\tilde C_i \cap \paren{X \setminus Z}} \cap  \paren{\tilde C_j \cap \paren{X \setminus Z}}
	= \emptyset \qquad \textrm{for each } i,j \in [k], i \neq j \mcom
\end{equation}
and,
\begin{equation}
\label{eq:muprox3}
\cup_{i \in [k]} \paren{ \tilde C_i \cap \paren{X \setminus Z}}
	= X \cap \paren{X \setminus Z} = \cup_{i \in [k]} C_i \mper
\end{equation}
\prettyref{eq:muprox1}, \prettyref{eq:muprox2} and \prettyref{eq:muprox3} imply the lemma.
\end{proof}
Therefore the \prettyref{lem:cip} implies the statement of the proposition, as in this setting, since $X \setminus Z = X$, \prettyref{lem:cip} implies that $\tilde C_i = C_i$.
\end{proof}

We are now ready to prove \prettyref{thm:mpr} and \prettyref{thm:mpr-gamma}
\begin{proof}[Proof of \prettyref{thm:mpr}]
Using \prettyref{prop:sampling} for a desired clustering $C_1,\ldots,C_k$, we get a set of points $\set{\tmu_1,\ldots,\tmu_k}$, such that
\[ \Pr{ \norm{\tmu_{i} - \mu_i} \leq \delta \diam{C_i} \textrm{ for each } i \in [k] }  \geq 1 - \beta_1 \mper \]
Using \prettyref{prop:amps}, we get that for each $i \in [k]$ $\diam{C_i} \leq 2\rij$ for all $j \in [k]$. Therefore this implies that
\[ \Pr{ \norm{\tmu_{i} - \mu_i} \leq 2 \delta \rij \textrm{ for all } j \in [k]\mcom \textrm{ for each } i \in [k] }  \geq 1 - \beta_1 \mper \]
This can be achieved by \prettyref{alg:stable-km}

Now, using \prettyref{prop:combined-sampling} we get that we can get an exact clustering. This step is achieved by \prettyref{alg:balance-means-known}.
\end{proof}

\begin{proof}[Proof of \prettyref{thm:mpr-gamma}]
	Using \prettyref{prop:sampling-mean} for a desired clustering $C_1,\ldots,C_k$, we get a set of points $\set{\tmu_1,\ldots,\tmu_k}$, such that \[ \norm{\tmu_i - \mu_i} \leq 2 \delta \rij \textrm{ for all } j \in [k] \textrm { and for each } i \in \brac{k}     \mper \]
	with constant probability.
	This step is achieved by \prettyref{alg:centers}.

	Now, using \prettyref{prop:combined-sampling} we get that we can get an exact clustering. This step is achieved by \prettyref{alg:gamma-means-known}.

	The running time of the algorithm is $\bigo{2^{\poly\paren{k / \e }} nd}$, and from equation \prettyref{eq:eps-val} and \prettyref{eq:alphadelta} we get that
	\[ \e \leq \paren{ \frac{ 2 (\alpha-1)^2 }{ 8 \alpha k \gamma } }^2 \mper \]
	Therefore, we get a running time of $\bigo{2^{\poly\paren{k \gamma / (\alpha-1) }} nd}$.
\end{proof}

\subsection{$\alpha$-Outliers Center Proximity}
In this section we assume that we can get access to a set of points $\set{\tmu_1,\ldots,\tmu_k}$ such that, $\norm{\tmu_i - \mu_i} \leq  \delta \diam{C_i} $ for all $i \in \brac{k}$, with probability some constant probability. We will show how to obtain such points in \prettyref{sec:balance-sampling}
\RestyleAlgo{boxruled}
\begin{algorithm}
\caption{$\alpha$-center proximal $k$-means clustering with balanced clusters and outliers}
\label{alg:balance-means-known-out}
\begin{algorithmic}[1]
\REQUIRE a list $\cL$ of $k$-tuples, where there exists at least one tuple $\paren{\tmu_1,\ldots,\tmu_k}$ such that for any $i \in \brac{k}$, \\  $\norm{\tmu_i - \mu_i} \leq \delta \diam{C_i} $ for all $j \in \brac{k}$, a number $\alpha >1$, $\omega > 0$, and $z > 0$
\ENSURE An $\alpha$-center proximal clustering where the size of each cluster is at least $\omega n / k$.
\STATE Initialize $(\tilde{\mu}_{1}^{(p)}, \tilde{\mu}_{2}^{(p)}, \dotsc, \tilde{\mu}_{k}^{(p)}) \leftarrow (\bar{0}, \bar{0}, \dotsc, \bar{0})$
\STATE $\text{cost} \leftarrow \sum_{i=1}^{n} \min_{1 \leq j \leq k} \norm{x_{i} - \tilde{\mu}_{j}^{(0)}}^{2}$
\FOR{ $\text{a tuple } \tilde \mu^{(q)} \in {{\cL}}$}
\FOR{$j=1 \text{ to } k$}
\STATE $C_{j}^{(q)} = \{x_{i} \suchthat \tilde{\mu}_{j}^{(q)}~ \text{is the nearest center to}~ x_{i}\}$
\ENDFOR
\STATE Remove the $z$ farthest points with respect to $(\tilde{\mu}_{1}^{(p)}, \tilde{\mu}_{2}^{(p)}, \dotsc, \tilde{\mu}_{k}^{(p)})$
\IF{$\sum_{i=1}^{n} \min_{1 \leq j \leq k} \norm{x_{i} - \tilde{\mu}_{j}^{(q)}}^{2} < \text{cost}$}
\IF{$C_{1}^{(q)},\ldots,C_{k}^{(q)} \text{ all have size at least } \frac{\omega n}{k} \text{ and satisfy } \alpha\text{-center proximity}$}
\STATE $(\tilde{\mu}_{1}^{(p)}, \tilde{\mu}_{2}^{(p)}, \dotsc, \tilde{\mu}_{k}^{(p)}) \leftarrow (\tilde{\mu}_{1}^{(q)}, \tilde{\mu}_{2}^{(q)}, \dotsc, \tilde{\mu}_{k}^{(q)}) $
\STATE cost $= \sum_{i=1}^{n} \min_{1 \leq j \leq k} \norm{x_{i} - \tilde{\mu}_{j}^{(q)}}^{2}$
\ENDIF
\ENDIF
\ENDFOR
\FOR{$i=1$ to $n$}
\STATE $\text{Label}(i) \leftarrow \underset{1 \leq j \leq k}{\text{argmin}} \norm{x_{i} - \tilde{\mu}_{j}^{(p)}}$
\ENDFOR
\end{algorithmic}
\end{algorithm}

We first note that, by our \prettyref{def:alpha_ocp}, outliers are the farthest most points in the data-set, where distance is measured from the set of optimal centers $\mu_1,\ldots,\mu_k$.
\begin{lemma}
  The \prettyref{def:alpha_ocp} implies that for any $q \in Z$
  \begin{equation}\label{eq:outlier_diam}
    \norm{q-\mu_j} \geq \alpha \frac{\diam{C_i} + \diam{C_j} }{4} \mcom \qquad \paren{\textrm{ for all } i , j \in [k]} \mcom
  \end{equation}
\end{lemma}
\begin{proof}
  From the \prettyref{def:alpha_ocp} we get that for any $p_i \in C_i$, $p_j \in C_j$ and $q \in Z$, we get that $\alpha \norm{p_i - \mu_i} \leq \norm{q - \mu_j}$ and  $\alpha \norm{p_j - \mu_j} \leq \norm{q - \mu_j}$.
  There exists at least one point $p_{i_1} \in C_i$ such that $\norm{p_{i_1} - \mu_i} \geq \diam{C_i}/2$, this is because the diameter is defined as $\max_{p_{i_1},p_{i_2} \in C_i}\norm{p_{i_1}-p{i_2}}$. We will prove this by contradiction.
  Suppose $\norm{p_1-p_2} = \diam{C_i}$, and $\norm{p_{i_1} - \mu_i} < \diam{C_i}/2$, and  $\norm{p_{i_2} - \mu_i} < \diam{C_i}/2$, then this contradicts the triangle inequality, as $\norm{p_1-p_2} > \norm{p_{i_1} - \mu_i} + \norm{p_{i_2} - \mu_i}$.

  Therefore we get that $\norm{q - \mu_j} \geq \alpha \diam{C_i}/2$ and $\norm{q - \mu_j} \geq \alpha \diam{C_j}/2$. Therefore the statement of the lemma follows.
\end{proof}

\begin{lemma} \label{lem:outlier_dist}
  For any $q \in Z$ and $x \in C_i$, we have that $\norm{x - \tmu_i} < \norm{q -  \tmu_j}$, for $i,j \in [k]$.
\end{lemma}
\begin{proof}
  \begin{align*}
  	\norm{x - \tmu_i} & \leq \norm{x - \mu_i} + \delta \diam{C_i} & \paren{\text{\prettyref{lem:cor-sampling}}} \\
    & < \frac{1}{\alpha} \norm{q - \mu_j} + \delta \diam{C_i} & \paren{\text{\prettyref{def:alpha_ocp}}} \\
    & = \paren{\norm{q - \mu_j} - \delta \diam{C_j}} + \paren{\frac{1}{\alpha} -1 } \norm{q - \mu_j} \\
    & \qquad \qquad	+ \delta \diam{C_i} + \delta \diam{C_j} \\
    & \leq \norm{q - \tmu_j} + \delta\paren{\diam{C_i} + \diam{C_j}} - \frac{\alpha - 1}{\alpha} \norm{q-\mu_j} & \paren{\text{\prettyref{lem:cor-sampling}}} \\
    & \leq \norm{q - \tmu_j} + \paren{\delta - \frac{(\alpha-1)}{4}}\paren{\diam{C_i} + \diam{C_j}} & \paren{\textrm{Using \prettyref{eq:outlier_diam}}} \\
    & < \norm{q -  \tmu_j} \mper & \paren{\textrm{Using \prettyref{eq:alphadelta}}}
  \end{align*}
\end{proof}

\begin{proof}[Proof of \prettyref{thm:mpr-out}]
  Using \prettyref{prop:sampling-mean} for a desired clustering $C_1,\ldots,C_k$, we get a set of points $\set{\tmu_1,\ldots,\tmu_k}$, such that \[ \norm{\tmu_i - \mu_i} \leq  \delta \diam{C_i} \textrm{ for all } i \in [k]  \mper \]
Using the \prettyref{lem:cip} we get that all the points are closer to their approximate means than other approximate means. This step is achieved by \prettyref{alg:stable-km}. Using the \prettyref{lem:outlier_dist} we get that even with respect to the approximate means, the outlier points are the farthest points of the data-set. Therefore we can remove the farthest $|Z|$ points, and get the desired clustering. This step is achieved by \prettyref{alg:balance-means-known-out}.
\end{proof}

\begin{proof}[Proof of \prettyref{thm:beta-mpr}]
The proof immediately follows from the \prettyref{prop:sampling-mean}. We think of $\beta$ points as outliers for the \prettyref{alg:stable-km}. With a constant probability, we get that one of the clusterings in the list produced by the algorithm corresponds to $(n-\beta)$ points being $\alpha$-center proximal.
\end{proof}

\section{Sampling}\label{sec:main-sampling}
\subsection{Balanced Cluster Assumption}\label{sec:balance-sampling}
\RestyleAlgo{boxruled}
\begin{algorithm}
\caption{List of $\alpha$-center proximal $k$-means with balanced clusters with outliers}
\label{alg:stable-km}
\begin{algorithmic}[1]
\REQUIRE a set of points $X = \{x_{1}, x_{2}, \dotsc, x_{n}\} \subseteq \R^{d}$, a positive integer $k \in \N$ , a number $\alpha >1$, $\omega >0$, and $z \geq 0$
\ENSURE A list $\cL$ of $k$-tuples, where each $k$-tuple $k$-mean points of a clustering.
\STATE $\delta \leftarrow (\alpha - 1)^{2}/8\alpha$
\STATE Initialize $(\tilde{\mu}_{1}^{(0)}, \tilde{\mu}_{2}^{(0)}, \dotsc, \tilde{\mu}_{k}^{(0)}) \leftarrow (\bar{0}, \bar{0}, \dotsc, \bar{0})$
\STATE $m \leftarrow (12/\omega)~ (128 k^{2} \alpha^{2}/(\alpha - 1)^{4} \beta_{1})~ \log (2k/\beta_{1})$
\STATE $ \cL \leftarrow \emptyset$
\STATE Pick an i.i.d. and uniform sample $Y$ of size $m$ from $X$.
\IF{z = 0}
\FOR{enumeration $p \in [T]$ over the set $T$ of all $k$ disjoint partitions $Y = Y_{1}^{(p)} \cup Y_{2}^{(p)} \dotsc \cup Y_{k}^{(p)}$}
\FOR{$j=1$ to $k$}
\STATE $\tilde{\mu}_{j}^{(p)} \leftarrow \frac{1}{\abs{Y_{j}^{(p)}}} \sum_{x_{i} \in Y_{j}^{(p)}} x_{i}$
\ENDFOR
\STATE $\cL = \cL \cup \set{\paren{\tilde{\mu}_{1}^{(p)},\ldots,\tilde{\mu}_{k}^{(p)}}}$
\ENDFOR
\ENDIF
\IF{$z >0$}
\FOR{enumeration $p \in [T]$ over the set $T$ of all $k+1$ disjoint partitions $Y = Y_{1}^{(p)} \cup Y_{2}^{(p)} \dotsc \cup Y_{k+1}^{(p)}$}
\FOR{$j=1$ to $k$}
\STATE $\tilde{\mu}_{j}^{(p)} \leftarrow \frac{1}{\abs{Y_{j}^{(p)}}} \sum_{x_{i} \in Y_{j}^{(p)}} x_{i}$
\ENDFOR
\STATE $\cL = \cL \cup \set{\paren{\tilde{\mu}_{1}^{(p)},\ldots,\tilde{\mu}_{k}^{(p)}}}$
\ENDFOR
\ENDIF
\end{algorithmic}
\end{algorithm}


The following proposition works for general $A_1,\ldots,A_k,Z$. For our case, $A_i$ corresponds to cluster $C_i$, and $Z$ corresponds to the set of outliers.
\begin{proposition}
\label{prop:sampling}
Fix $\delta, \beta_1 \in (0,1)$.
Let $X$ be a set of points in $\R^d$,
partitioned into $k+1$ sets
$A_1, \ldots, A_k, Z$ such that for each $i$, $\Abs{A_i} \geq \wmin \Abs{X}/k$.
Let $\mu_i$ denote the mean of $A_i$, i.e., $\mu_i \defeq \paren{\sum_{v \in A_i} v}/\Abs{A_i}$.
There exists a randomized algorithm which outputs the set
$\set{ (\tmuu{1}{p}, \ldots, \tmuu{k}{p}) : p \in [T]}$ satisfying
\[ \Pr{ \exists p \in [T] \textrm{ such that }
\norm{\tmuu{i}{p} - \mu_i} \leq \delta \diam{A_i} \textrm{ for each } i \in [k] }  \geq 1 - \beta_1  \mcom \]
in $\bigo{T}$  iterations where  $T < (k+1)^{\frac{16k^2}{\delta^2 \beta_1} \log\paren{\frac{2k}{\beta_1}}}$.
\end{proposition}

We begin by recalling the following lemma, which is implicit in Theorem 2 of \cite{Barman15}, which says that with constant probability, one can get close to the mean of a set of points with bounded diameter, by randomly sampling a constant number of points.
\begin{lemma}
\label{lem:sampling3}
Fix a set of elements $A$ and a set of parameters
	$\delta, \beta_3 \in (0,1)$, and let $l_0 = 1/\paren{\delta^2 \beta_3}$.
Let $l \in \Z$ be any number such that $l \geq l_0$, and
let $y_1, \ldots y_l$ be $l$ independent and uniformly random samples from $A$. Then,
	\[ \Pr{ \norm{\frac{1}{l} \sum_{i = 1}^l y_i - \frac{1}{\Abs{A}} \sum_{a \in A} a  } \leq
	\delta \diam{A} } \geq 1 - \beta_3 \mper  \]
\end{lemma}

	\begin{Claim}\label{claim:approx}
		\[  \E { \norm{\frac{1}{l} \sum_{i = 1}^l y_i - \frac{1}{\Abs{A}} \sum_{a \in A} a  }^2} \leq \diam{A}^2/l \]
	\end{Claim}
	\begin{proof}
		(We are reproducing the proof of \cite{Barman15} for the sake of completeness.)
	\begin{align*}
		\E { \norm{\frac{1}{l} \sum_{i = 1}^l y_i - \frac{1}{\Abs{A}} \sum_{a \in A} a  }^2} & = \frac{1}{l^2} \E \norm{ \sum_{i=1}^l \left( \frac{1}{\Abs{A}} \sum_{a \in A} a - y_i \right)}^2 \\
		& = \frac{1}{l^2} \E \inprod{\sum_{i=1}^l \paren{\frac{1}{\Abs{A}} \sum_{a \in A} a-y_i }, \sum_{j=1}^l \left( \frac{1}{\Abs{A}} \sum_{a \in A} a-y_j \right) } \\
		& = \frac{1}{l^2} \sum_{i,j=1}^l \E \inprod{\left(  \frac{1}{\Abs{A}} \sum_{a \in A} a-y_i \right),  \left( \frac{1}{\Abs{A}} \sum_{a \in A} a-y_j \right) }  \\
		& = \frac{1}{l^2} \sum_{i=1}^l \E \inprod{\left(  \frac{1}{\Abs{A}} \sum_{a \in A} a-y_i \right),  \left( \frac{1}{\Abs{A}} \sum_{a \in A} a-y_i \right) }  &  \paren{\textrm{Since $y_i$ and $y_j$ are independent}} \\
		& = \frac{1}{l^2} \sum_{i=1}^l \paren{\norm{ \frac{1}{\Abs{A}} \sum_{a \in A} a}^2 - 2\norm{ \frac{1}{\Abs{A}} \sum_{a \in A} a}^2 + \E \norm{y_i}^2} & \paren{\E y_i =  \frac{1}{\Abs{A}} \sum_{a \in A} a} \\
		& \leq \frac{1}{l^2} l \diam{A}^2 & \paren{\E \norm{y_i}^2 - \norm{x}^2 \leq \diam{A}^2 } \\
		& = \frac{\diam{A}^2}{l} \\
	\end{align*}
	\end{proof}
	\begin{proof}[Proof of \prettyref{lem:sampling3}]
	Let the random variable $Z \defeq   \norm{\frac{1}{l} \sum_{i = 1}^l y_i - \frac{1}{\Abs{A}} \sum_{a \in A} a  }^2 $. We know from the  \prettyref{claim:approx} that  $\E(Z) \leq \frac{\diam{A}^2}{l} $. Therefore using Markov's inequality, we get
	\[ \Pr{ Z \geq \frac{\diam{A}^2}{l \beta_3}} \leq \frac{\E Z}{\frac{\diam{A}^2}{l \beta_3}} \leq  \frac{\frac{\diam{A}^2}{l}}{\frac{\diam{A}^2}{l \beta_3}} = \beta_3\]
	\[ \Pr{\Abs{\sqrt{Z}} \geq \frac{\diam{A}}{\sqrt{l} \sqrt{\beta_3}}} \leq \beta_3 \]
	Therefore choosing $\delta = \frac{1}{\sqrt{\beta_3}\sqrt{l}}$, we get the statement of \prettyref{lem:sampling3}.
\end{proof}

The above lemma helps us determine the number of points to be taken from each of the clusters to get close to the mean. The next helps us determine the number of points to sample uniformly at random from the data-set, such that we get at least some fixed number of points from each balanced clusters (with a balance parameter $\omega$), with constant probability.
\begin{lemma}
\label{lem:sampling2}
Fix $l_0 \in \Z$ and $\beta_2 \in (0,1/2]$.
Let $X = \set{x_1, \ldots, x_n}$ be a set of $n$ distinct items partitioned into $k+1$ sets
$A_1, \ldots, A_k, Z$ such that $\Abs{A_i} \geq \wmin n /k$ for each $i \in [k]$.
Let $m = (8/\wmin) k l_0 \log \paren{ k/ \beta_2}$, and let $Y = \set{y_1, \ldots, y_m}$ be a set of $m$ independently
and uniformly randomly chosen elements from $X$ (with repetition).
Then,
	\[ \Pr{ \Abs{Y \cap A_i} \geq l_0 \textrm{ for each } i \in [k]} \geq 1 - \beta_2 \mper \]
\end{lemma}

\begin{proof}
Let $p_i \defeq \Abs{A_i}/n$ and let $Y_i \defeq \Abs{Y \cap A_i}$.
For our choice of parameters, we get that
	\[ \frac{p_i m}{2} \geq \frac{\wmin m}{2k} \geq l_0 \qquad \textrm{ and } \qquad
	\frac{p_i m}{8}  \geq  \log \frac{k}{\beta_2}	\mper \]
Therefore, using the Chernoff bound on $Y_i$, we get that
\[  \Pr{Y_i \geq l_0} \geq \Pr{ Y_i \geq  p_i m /2 } \geq 1 - e^{- p_i m/8} \geq 1 - \frac{\beta_2}{k} \mper \]
Using a union bound over each $i \in [k]$, we get
\[ \Pr{ \Abs{Y \cap A_i} \geq l_0 \textrm{ for each } i \in [k]}
	\geq 1 - k \frac{\beta_2}{k} = 1 - \beta_2 \mper \]
\end{proof}

We are now ready to prove \prettyref{prop:sampling}.
\begin{proof}[Proof of \prettyref{prop:sampling}]
Set $\beta_2 = \beta_1/2$, $\beta_3 = \beta_1/(2k)$,
set $l_{0}$ from the guarantee in
\prettyref{lem:sampling3}, and set $m$ from the guarantee in \prettyref{lem:sampling2}.
Let $Y = \set{y_1, \ldots, y_m}$ be a set of $m$ independently
and uniformly randomly chosen elements from $X$ (with repetition)
It is easy to verify that the points
in $Y \cap A_i$ are a uniformly random sample from $A_i$.

\begin{equation}
\label{eq:sampling4}
	\Pr{ y_j = a \given y_j \in A_i } = \frac{1}{ \Abs{A_i}} \qquad \textrm{for any } i \in [k], a \in A_i \textrm{ and } j \in [m]
\end{equation}

Let $\tmuu{i}{*} \defeq \paren{\sum_{j \in Y \cap A_i} y_j}/{\Abs{Y \cap A_i}}$.
If we enumerate all the $k+1$ partitions on $Y$, then one of the partitions will be
$\set{Y \cap A_1, \ldots, Y \cap A_k, Y \cap Z}$.
	For a fixed index $i \in [k]$, using \prettyref{lem:sampling2} and \prettyref{eq:sampling4}
with $A = A_i$ and $Y \cap A_i$ as the set of random samples from $A_i$, we get that
\[ \Pr{ \Abs{Y \cap A_i} \geq l_0 \textrm{ for each } i \in [k]} \geq 1 - \beta_2 = 1 - \frac{\beta_1}{2} \mper \]
In this case, using \prettyref{lem:sampling3},
we get that
\[ \Pr{ \norm{\tmuu{i}{*} - \mu_i} \leq \delta \diam{A_i} \textrm{ for each } i \in [k]} \geq
	1 - k \beta_3 = 1 - \frac{\beta_1}{2} \mper \]
Using a union bound over these two events, we get that
\[ \Pr{ \norm{\tmuu{i}{*} - \mu_i} \leq \delta \diam{A_i} \textrm{ for each } i \in [k]} \geq
	1 - \frac{\beta_1}{2} - \frac{\beta_1}{2} = 1 - \beta_1 \mper \]

The running time of this algorithm is dominated by the time required to enumerate all the $k+1$ partitions of set of cardinality $m$, and computing the means of those partitions.
\end{proof}

\subsection{Balanced Mean Distance Assumption}
Assume that the unknown $\alpha$-center proximal $k$-means clustering of lowest cost is $\set{\Opt_1,\ldots,\Opt_k}$, with means $\set{\mu_1,\ldots,\mu_k}$ respectively. In this section we assume that the ratio of the maximum pairwise distance between the means to the minimum pairwise distance between the means is bounded by a factor $\gamma^*$. We assume that we are given an upper bound $\gamma$ on $\gamma^*$. More formally, we are given a $\gamma$, such that
\[ \gamma \geq  \frac{\max_{i,j} \norm{\mu_i-\mu_j}}{\min_{i \neq j} \norm{\mu_i-\mu_j}}  \mper \]

\begin{proposition}
\label{prop:sampling-mean}
Fix $\delta \in (0,1)$.
Let $X$ be a set of points in $\R^d$,
partitioned into $k$ sets
$C_1, \ldots, C_k$, and let $\mu_i$ denote the mean of $C_i$, i.e., $\mu_i \defeq \paren{\sum_{v \in C_i} v}/\Abs{C_i}$, such that $\gamma =  {\max_{i,j} \norm{\mu_i-\mu_j}} \bigg/ {\min_{i \neq j} \norm{\mu_i-\mu_j}}$. Let $\e = \paren{\frac{2 \delta k}{\gamma (1+k)}}^2$.
The \prettyref{alg:centers} constructs $\bigo{2^{\poly\paren{\frac{k}{\e}}}}$ $k$-tuples. With constant probability, there exists at least one $k$-tuple $(\tmu_1, \ldots, \tmu_k)$ satisfying
\[ \norm{\tmu_i - \mu_i} \leq 2 \delta \rij \textrm{ for all } j \in [k] \textrm { and for each } i \in \brac{k}     \mper \]
Moreover, the algorithm runs in time $ \bigo{2^{\poly\paren{\frac{k}{\e}}}nd} $.
\end{proposition}.
The proof of the \prettyref{prop:sampling-mean} is similar to the proof of Lemma 3.3 of \cite{Ding15}, with minor modifications, keeping our application in mind.

Using \prettyref{prop:geom} we define $\radm{X} \defeq \frac{\alpha}{\alpha^2-1} \min_{i \neq j} \norm{\mu_i-\mu_j}$. The proposition implies that the radius of the largest $C_{i,j}$ ball (for some $i,j \in \brac{k}$) is at most $\gamma \radm{X}$, ie.,
\begin{equation}\label{eq:radmin}
   \frac{\rij}{\gamma} \leq  \radm{X} \leq \rij \qquad \text{ for } i,j \in [k] \mper
\end{equation}

\begin{proposition}
  Let $R \defeq \max_{i,j \in \brac{n}} \norm{x_i - x_j}$, and let $\radm{X} \defeq \frac{\alpha}{\alpha^2-1} \min_{i \neq j} \norm{\mu_i-\mu_j}$. Then,
  \begin{equation}
    \radm{X} \in \Brac{ \paren{\frac{\alpha}{(\alpha+1)^2}}\frac{R}{\gamma}, \paren{\frac{\alpha}{\alpha^2-1}}\frac{R}{\gamma}} \mper
  \end{equation}

\end{proposition}
\begin{proof} It is easy to see that
\[ R \geq \max_{i,j} \norm{\mu_i - \mu_j} \mcom \qquad \text{for } i,j \in \brac{k} \mcom  \]
and from \prettyref{prop:geom} that (points lie in two balls of radius $\rij$ which are at a distance $d_{i,j}$ apart)
\[ R \leq \max_{i,j} \paren{ 4 \rij + d_{i,j} } =  \paren{\frac{\alpha+1}{\alpha-1}}\max_{i,j} \norm{\mu_i - \mu_j} \mcom  \qquad \text{for } i,j \in \brac{k} \mper \]
Using the above two equations and the fact that $\gamma =  {\max_{i,j} \norm{\mu_i-\mu_j}} \bigg/ {\min_{i \neq j} \norm{\mu_i-\mu_j}}$, we get that
\[ \min_{i \neq j} \norm{\mu_i-\mu_j} \in \Brac{\paren{\frac{\alpha-1}{\alpha+1}}\frac{R}{\gamma} \mcom \frac{R}{\gamma}} \mcom \qquad \text{for } i,j \in \brac{k} \mcom \]
which implies that (using \prettyref{prop:geom})
\begin{equation*}
  \radm{X} \in \Brac{ \paren{\frac{\alpha}{(\alpha+1)^2}}\frac{R}{\gamma}, \paren{\frac{\alpha}{\alpha^2-1}}\frac{R}{\gamma}} \mper
\end{equation*}
\end{proof}

\RestyleAlgo{boxruled}
\begin{algorithm} \label{alg:centers}
\caption{Algorithm Peeling-and-Enclosing}
\textbf{Input:} $X = \set{x_1,\ldots,x_n}$ in $\R^d$, $k \geq 2$, $\alpha$, $\gamma$. \\
\textbf{Output:} A list $\cL$ containing $k$-tuples, where a $k$-tuple contains $k$ mean points.\\
\begin{enumerate}
\item Set $\e = \paren{\frac{(\alpha-1)^2 k }{4 \alpha \gamma (1 + k)}}^2$.
\item For $i = 0 \text{ to } \log_{(1+\e)}\paren{\frac{\alpha+1}{\alpha-1}}$ do
  \begin{enumerate}
  \item $\zeta = (1+\e)^i \paren{\frac{\alpha}{(\alpha+1)^2}}\frac{R}{\gamma}$.
  \item Run \textbf{Algorithm Peeling-and-Enclosing-Tree}.
  \item Let $\cT_i$ be the output tree.
  \end{enumerate}
\item For each root-to-leaf path of every $\cT_i$, build a $k$-tuple candidate using the $k$ points associated with the path.
\item Append the $k$-tuple to the list $\cL$.
\end{enumerate}
\end{algorithm}

The \prettyref{alg:tree} is almost same as the Algorithm Peeling-and-Enclosing-Tree of \cite{Ding15}, with a minor variation in step 2(b).
\begin{algorithm} \label{alg:tree}
\caption{Algorithm Peeling-and-Enclosing-Tree}
\textbf{Input:} $\zeta$ and an instance of $k$-means $X$.\\
\textbf{Output:} A tree $\cT$.
\begin{enumerate}
\item Initialize $\cT$ as a single root node $v$ associated with no point.
\item Recursively grow each node $v$ in the following way
  \begin{enumerate}
  \item If the height of $v$ is already $k$, then it is a leaf.
  \item Otherwise, let $j$ be the height of $v$. Build the radius candidate set $\cR =  \set{\frac{1+l\frac{\e}{2}}{2(1+\e)}j \sqrt{2}\sqrt{\e}\zeta\gamma \middle| 0\leq l \leq 4+\frac{2}{\e}}$\\ For each $r \in \cR$, do
    \begin{enumerate}
    \item Let $\set{p_{v_1},\ldots,p_{v_j}}$ be the $j$ points associated with nodes on the root-to-v path.
    \item For each $p_{v_l}, 1 \leq l \leq j$, construct a ball $\cB_{j+1,l}$ centered at $p_{v_l}$ and with radius $r$.
    \item Take a random sample from $\paren{X \setminus \cup_{l=1}^j \cB_{j+1,l}}$ of size $s = \frac{8k^3}{\e^9} \ln \frac{k^2}{\e^6}$. Compute the mean points of all subset of the sample, and denote them by $\Pi = \set{\pi_1,\ldots,\pi_{2^s-1}}$.
    \item For each $\pi_i \in \Pi$, construct a simplex using $\set{p_{v_1},\ldots,p_{v_j},\pi_i}$ as its vertices. Also construct another simplex using $\set{p_{v_1},\ldots,p_{v_j}}$ as its vertices. For each simplex, build a grid with size $\bigo{\paren{32j/\e}^j}$ inside itself and each of its $2^j$ possible degenerated sub-simplices.
    \item In total, there are $2^{s+j}\paren{32j/\e}^j$ grid points inside the $2^s$ simplices. For each grid point, add one\\ child to $v$, and associate it with the grid point.
    \end{enumerate}
  \end{enumerate}
\item Output $\cT$.
\end{enumerate}
\end{algorithm}

We note a set of preliminary lemmas and definition which we will need for the proof of \prettyref{prop:sampling-mean}.
\begin{definition}[Simplex]
  A $k$-simplex is a $k$-dimensional polytope which is the convex hull of its $k + 1$ vertices. More formally, suppose the $k + 1$ points
  $u_{0},\dots ,u_{k}\in \mathbb {R} ^{k}$ are affinely independent, Then, the simplex determined by them is the set of points
  \[ \cV = \left\{\theta_0 u_0 + \dots +\theta_k u_k ~\bigg|~ \sum_{i=0}^{k} \theta_i=1 \mbox{ and } \theta_i \ge 0 \mbox{ for all } i \in [k]\right\} \mper \]
\end{definition}
\begin{lemma}[Lemma 1, \cite{Inaba94}] \label{lem:inaba1}
  Let $S$ be a set of $n$ points in $\R^d$, $T$ be a randomly selected subset of size $t$ from $S$, and $\mu(S), \mu(T)$ be the mean points of $S$ and $T$ respectively. With probability $1-\eta$, $\norm{\mu(S)-\mu(T)}^2 \leq \frac{1}{\eta t} \sigma^2$,
  where $\sigma^2 = \frac{1}{n} \sum_{s \in S} \norm{s-\mu(s)}^2$ and $0 \leq \eta \leq 1$.
\end{lemma}

\begin{lemma}[Lemma 4, \cite{Ding14}] \label{lem:ding14}
  Let $\Gamma$ be a set of elements, and $S$ be a subset of $\Gamma$ with $\frac{|S|}{|\Gamma|}=\rho$ for some $\rho \in (0,1)$.
  If we randomly select $\frac{t \ln \frac{t}{\eta}}{\ln(1+\rho)} = \bigo{\frac{t}{\rho} \ln \frac{t}{\eta}}$ elements from $\Gamma$, then with probability at least $1-\eta$, the sample contains $t$ or more elements from $S$ for $0 < \eta < 1$ and $t \in \mathbb Z^+$.
\end{lemma}

\begin{lemma}[Lemma 2.3 (Simplex Lemma II), \cite{Ding15}] \label{lem:simplex}
  Let $Q$ be a set of points in $\R^d$ with a partition of $Q = \cup_{l=1}^{j} Q_l$ and $Q_{l_1} \cap Q_{l_2} = \emptyset$ for any $ l_1 \neq l_2$.
  Let $o$ be the mean point of $Q$, and $o_l$ be the mean point of $Q_l$ for $1 \leq l \leq j$.
  Let $\sigma^2 = \frac{1}{|Q|} \sum_{q \in Q} \norm{q-o}^2$. Let $\{o_1',\ldots,o_j'\}$ be $j$ points in $\R^d$ such that $\norm{o_l-o_l'} \leq L$ for $1 \leq l \leq j, L>0$, and $\simplex$ be the simplex determined by $\{o_1',\ldots,o_j' \}$.
  Then for any $0 < \e \leq 1$, it is possible to construct a grid of size $\bigo{(8j/\e)^j}$ inside $\simplex$ such that at least one grid point $\tau$ satisfies the inequality $\norm{\tau-o} \leq \sqrt{\e}\sigma + (1+\e)L$.
\end{lemma}

\begin{lemma}[Lemma 2.2, \cite{Ding15}] \label{lem:ding15}
Let $Q$ be a set of points in $\R^d$, and $Q_1$ be its subset containing $\rho \abs{Q}$ points for some $0 < \rho \leq 1$. Let $o$ and $o_1$ be the mean points of $\Q$ and $Q_1$, respectively.
Then $\norm{o-o_1} \leq \sqrt{\frac{1-\rho}{\rho}}\sigma$, where $\sigma^2 = \frac{1}{\abs{Q}} \sum_{q \in Q} \norm{q-o}^2$.
\end{lemma}

\paragraph{Notations:} Let $\mathcal{OPT}=\set{\Opt_1,\ldots,\Opt_k}$ be the $k$ unknown optimal clusters for the lowest cost $\alpha$-center proximal $k$-means objective, with means $\mu_j$. W.l.o.g. we assume that $|\Opt_1|\geq \ldots \geq |\Opt_k|$.
We define $\sigma_{j}^2 \defeq \frac{1}{\abs{\Opt_j}}\sum_{p \in \Opt_j} \norm{p-\mu_{j}}^2$. Let $\lambda_j \defeq |\Opt_j|/n$.

The following lemma is similar to the Lemma 3.3 of \cite{Ding15} with minor modifications.
\begin{lemma}\label{lem:closepts}
  Among all the points generated by the \prettyref{alg:tree}, with constant probability, there exists at least one tree, $\mathcal T_i$, which has a root-to-leaf path with each of its nodes $v_j$ at level $j, (1 \leq j \leq k)$ associating with a point $p_{v_j}$ and satisfying the inequality
  \begin{equation} \label{eq:epsgamma}
    \norm{p_{v_j}-\mu_j} \leq \e \gamma \radm{X} + (1+\e) j \sqrt{\e} \gamma \radm{X} \mper
  \end{equation}
\end{lemma}
\paragraph{Algorithm and Proof Overview:} We will give a high level idea of the algorithm and the proof. At each searching step, the algorithm performs a `sphere peeling' and `simplex enclosing' step, to generate $k$ approximate mean points for the clusters. Initially the algorithm uses a random sampling technique to find an approximate mean $p_{v_1}$ for $\Opt_1$. This can be done as $\frac{\abs{\Opt_1}}{n} \geq 1/k$, and hence we can sample.For some $j\geq 1$, suppose that at the $(j+1)^{th}$ iteration, the algorithm already has approximate mean points mean points $p_{v_1},\ldots,p_{v_j}$ for $\Opt_1,\ldots,\Opt_j$.
It is not clear how to distinguish points which belong to $\Opt_1,\ldots,\Opt_j$ from those which belong to $\Opt_{j+1}$.
Also, the number of points in the cluster $\Opt_{j+1}$ could be small, it is tough to obtain a significant fraction of such points using random sampling. Therefore, the idea used is to seperate the points in $\Opt_{j+1}$ using $j$ peeling spheres, $\cB_{j+1,1},\ldots,\cB_{j+1,j}$, centered at the $j$ approximate mean points respectively and with radius approximately being $\radm{X}$. Note that $\cB_{j+1,1},\ldots,\cB_{j+1,j}$ can have some points from $\Opt_{j+1}$.
Let $P_{j+1}$ be the set of unknown points in $\Opt_{j+1} \setminus \paren{ \cup_{l=1}^j \cB_{j+1,l}}$.
The algorithm considers two cases, a) $\abs{P_{j+1}}$ is large and b) $\abs{P_{j+1}}$ is small. For the case a) when $\abs{P_{j+1}}$ is large, we can sample points from $P_{j+1}$ using random sampling, and get an approximate mean $\pi$ of $P_{j+1}$, and then construct a simplex determined by $\pi,p_{v_1},\ldots,p_{v_j}$ to contain the $(j+1)^{th}$ mean point, using \prettyref{lem:simplex}.
This is because, $\Opt_{j+1} \cap \cB_{j+1,l}$, $l \in [j]$ can be seen as a partition of $\Opt_{j+1}$ whose approximate mean is $p_{v_l}$, thus the simplex lemma II applies. For case b) where $\abs{P_{j+1}}$ is small, it directly constructs the simplex determined by $p_{v_1},\ldots,p_{v_j}$, and searches for the approximate mean point of $\Opt_{j+1}$ in the grid.
This follows because it can be shown that $\Opt_{j+1} \cap \cB_{j+1,l}$, $l \in [j]$ can be seen as a partition of $\Opt_{j+1}$ whose approximate mean is $p_{v_l}$, and from the \prettyref{lem:ding15}. The \prettyref{lem:ding15} roughly says that even if we remove a small number of points from a cluster, its new mean remains close to the original mean.

\begin{proof}[Proof of \prettyref{lem:closepts}:]
  Let $\cT_i$ be the tree generated by the \prettyref{alg:tree} when
  \[ \zeta \in \left [\radm{X}, (1+\e)\radm{X} \right] \mper \]
  We will prove this lemma by induction.

  \textbf{Base Case:} For $j=1$, we have $\lambda_1 \geq \frac{1}{k}$. Therefore through random sampling (\prettyref{lem:ding14}), we can find a point $p_{v_1}$, which is close to $\mu_1$ (\prettyref{lem:inaba1}).
  We get that $\norm{p_{v_1} - \mu_1} \leq \e\sigma_1 \leq \e\radm{X} + (1+\e) \gamma \sqrt{\e} \radm{X}$. Hence the base case holds.

  \textbf{Induction Step:} We assume that there is a path in $\cT_i$ from root to the $(j-1)$-level, such that for each $1 \leq l \leq (j-1)$, the level-$l$ node $v_l$ on the path associated with a point $p_{v_l}$ satisfying the inequality
  \[ \norm{p_{v_l} - \mu_l} \leq \e \gamma \radm{X} + (1+\e) l \sqrt{\e} \gamma \radm{X} \mper\]
  Now we need to show this for the $j$-level, i.e., we need to show that there exists at least one child $v_j$ of $v_{j-1}$, such that the associated point $p_{v_j}$ satisfies the inequality
  \[ \norm{p_{v_j} - \mu_j} \leq \e \gamma \radm{X} + (1+\e) j \sqrt{\e} \gamma \radm{X} \mper \]
  First we make the following claim. The claim is a slight modification of Claim 2 of \cite{Ding15}. We will prove it in \prettyref{sec:claim-pf}.
  \begin{Claim}\label{claim:radius}
    In the set of radius candidates in the algorithm, there exists one value $r_j \in \mathcal R$, such that
    \[ r_j \in \Brac{j\sqrt{\e} \gamma \radm{X}, \paren{1+\frac{\e}{2}}j\sqrt{\e} \gamma \radm{X} } \mper \]
  \end{Claim}
  Now we construct $(j-1)$ peeling spheres $\set{\cB_{j,1},\ldots,\cB_{j,j-1}}$. For each $1 \leq l \leq j-1$, $\cB_{j,l}$ is centered at $p_{v_l}$ with radius $r_j$. Next we make the following claim. The proof claim is similar to the proof of Claim 3 of \cite{Ding15}, adapted to our setting. We will prove it in \prettyref{sec:claim-pf}.
  \begin{Claim}\label{claim:peeling}
    For each $1 \leq l \leq j-1$, $\abs{\Opt_l \setminus \paren{\cup_{w=1}^{j-1} \cB_{j,w}}} \leq \frac{4\lambda_j n}{\e}$.
  \end{Claim}
  \prettyref{claim:peeling} shows that $\abs{\Opt_l \setminus \paren{\cup_{w=1}^{j-1} \cB_{j,w}}}$ is bounded for $1 \leq l \leq j-1$, which helps us to find the approximate mean of $\Opt_j$.
  $\Opt_j$ is divided into $j$ subsets, $(\Opt_j \cap \cB_{j,1}), \ldots, (\Opt_j \cap \cB_{j,j-1})$, and $\Opt_j \setminus \paren{\cup_{w=1}^{j-1} \cB_{j,w}}$.
  Let $P_l$ denote $\Opt_j \cap \cB_{j,l}$ for $1 \leq l \leq j-1$, and $P_j$ denote $\Opt_j \setminus \paren{\cup_{w=1}^{j-1} \cB_{j,w}}$, and $\tau_l$ denote mean point of $P_l$ for $1 \leq l \leq j$.
  We can assume that $\set{P_l \middle|1 \leq l \leq j }$ are pairwise disjoint. If not, then arbitrarily assign points to either of the peeling spheres which intersect in $\Opt_j$.

  We now have two cases: (a) $\abs{P_j} \geq \e^3\frac{\lambda_j}{j}n$, and (b) $\abs{P_j} < \e^3\frac{\lambda_j}{j}n$. We show that \prettyref{alg:tree} can obtain an approximate mean for $\Opt_j$ by using the \prettyref{lem:simplex}, for both the cases.

  \textbf{Case (a):} By \prettyref{claim:peeling}, and using that fact that $\lambda_l \leq \lambda_j$ for $l > j$, we know that
  \begin{align*}
    \frac{\abs{P_j}}{\sum_{1 \leq l \leq k}\abs{\Opt_l \setminus \paren{\cup_{w=1}^{j-1} \cB_{j,w}}}} & \geq \frac{\frac{\e^3}{j}\lambda_j}{\frac{4(j-1)\lambda_j}{\e}+ \lambda_j+(k-j)\lambda_j} \mcom \\
    & \geq \frac{\e^4}{8kj} \mcom \\
    & \geq \frac{\e^4}{8k^2} \mper
  \end{align*}
  This means that $P_j$ is large enough, compared to the points outside the peeling spheres. Hence we can use random sampling technique to obtain an approximate mean point $\pi$ for $P_j$ in the following way.
  First we set $t= \frac{k}{\e^5}$, $\eta = \frac{\e}{k}$, and take sample of size $\frac{8k^3}{\e^9} \ln \frac{k^2}{\e^6}$. By \prettyref{lem:ding14} we know that with probability $1-\frac{\e}{k}$, the sample contains $\frac{k}{\e^5}$ points from $P_j$.
  Let $\pi$ be the mean of the $\frac{k}{\e^5}$ points sampled from $P_j$, and let $a^2$ be the variance of $P_j$. By \prettyref{lem:inaba1} we know that with probability $1-\frac{\e}{k}$, $\norm{\pi - \tau_j}^2 \leq \e^4a^2$.
  Also, since $\frac{\abs{P_j}}{\abs{\Opt_j}} = \frac{\abs{P_j}}{\lambda_j n} \geq \frac{\e^3}{j}$, we have $a^2 \leq \frac{\abs{\Opt_j}}{\abs{P_j}}\sigma_j^2 \leq \frac{1}{\e^3}\sigma_j^2$.
  This upper bound on $a^2$ follows because $P_j \subset C_j$. We are summing the distance square over all the elements in $C_j$, and then dividing by $\abs{P_j}$.
  Thus, $\norm{\pi-\tau_j}^2 \leq \e j \sigma_j^2 \leq \e j \gamma^2 \radm{X}^2$.

  After obtaining the point $\pi$, we can use the \prettyref{lem:simplex} to find a point $p_{v_j}$ satisfying the condition of $\norm{p_{v_j}-\mu_j} \leq \e \gamma \radm{X} + (1+\e)j \sqrt{\e} \gamma \radm{X}$. This is true because of the following.
  First we construct the simplex $\simplex_{(a)}$ with vertices $\set{p_{v_1},\ldots,p_{v_{j-1}},\pi}$, which is $o_1',...o_{j-1}',o_j'$ in the lemma.
  Note that $\Opt_j$ which is $Q$ in the lemma is partitioned by the peeling spheres into $j$ disjoint subsets $P_1,\ldots,P_j$, which is $Q_1,\ldots,Q_j$ in the lemma. Each $P_l$ $(1 \leq l \leq j-1)$ locates inside $\cB_{j,l}$, which implies that $\tau_l$ (mean of $P_l$), which is $o_l$ in the lemma, is also inside ${\cB_{j,l}}$. Further by \prettyref{claim:radius} we have for $1 \leq l \leq j-1$:
  \begin{equation} \label{eq:mean1}
   j\sqrt{\e} \gamma \radm{X} \leq r_j \leq \paren{1+\frac{\e}{2}}j\sqrt{\e} \gamma \radm{X} \mper
  \end{equation}
  We also have from above that (recall that $\tau_l$ is the mean of $P_l$)
  \begin{equation} \label{eq:mean2}
    \norm{\pi-\tau_j} \leq \sqrt{\e j} \gamma \radm{X} \mper
  \end{equation}
  By \prettyref{eq:mean1} and \prettyref{eq:mean2} we know that if we set the value of $L$ and $\e$ (in \prettyref{lem:simplex}) to be $L = \paren{1+\frac{\e}{2}}j\sqrt{\e} \gamma \radm{X}$ and $\e$ to be $\e_0 = \e^2/4$,
  by \prettyref{lem:simplex} we can construct a grid inside the simplex $\simplex_{(a)}$ with size $\bigo{(8j/\e_0)^j}$ to ensure existence of one grid point $\tau$ satisfying that inequality
  \[ \norm{\tau -\mu_j} \leq \sqrt{\e_0}\sigma_j + (1+\e_0)L \leq \e \gamma \radm{X} + (1+\e)j \sqrt{\e} \gamma \radm{X} \mper \]
  Hence we can use $\tau$ as $p_{v_j}$, and the induction step holds.

  \textbf{Case (b):} We can use the \prettyref{lem:simplex} to find an approximate mean point. This is true because of the following. We construct a simplex $\simplex_{(b)}$ with vertices $\set{p_{v_1},\ldots,p_{v_{j-1}}}$. Since $\abs{P_j}$ is small, the mean points of $\Opt_j \setminus P_j$ and $\Opt_j$ are very close to each other (\prettyref{lem:ding15}).
  Thus we can ignore $P_j$ and consider only $\Opt_j \setminus P_j$. Here the value of $\rho$ in \prettyref{lem:ding15} is $(1 - \e^3/j)$. Thus, the $(j-1)$ dimensional simplex will approximate the mean of $P_j$ well (by \prettyref{lem:simplex}). Here the value of $L$ and $\e$ is same as in the case (a). Therefore, the induction step holds for this case as well.

  Since \prettyref{alg:tree} executes every step in the above discussion, the induction step, as well as the lemma, is true.

\end{proof}

\begin{proof}[Proof of \prettyref{prop:sampling-mean}]
  From the \prettyref{lem:closepts} we know that with constant probability our algorithm can find a point $\tmu_j$, $1 \leq j \leq k$ such that
  \[\norm{\tmu_j-\mu_j} \leq \e \gamma \radm{X} + (1+\e) j \sqrt{\e} \gamma \radm{X} \mper \]
  Therefore we now calculate the value of $\e$, the success probability and the running time.
  \paragraph{$\e$ Value:} From \prettyref{eq:epsgamma} (which implies that $\norm{\tmu_j - \mu_j} \leq \e \gamma \radm{X} + (1+\e) k \sqrt{\e} \gamma \radm{X}$ ) and \prettyref{eq:radmin} (which gives a bound on $\radm{X}$).
  From proof of \prettyref{claim:peeling} we know that $\e \leq \frac{1}{4k^2}$.
  We want to show that $\norm{\tmu_j - \mu_j} \leq 2\delta\rij $. Therefore, we get that for $i , j \in [k]$
  \begin{align}
    (\e+(1+\e)k\sqrt{\e})\gamma \rij & \leq 2\delta \rij \mcom \nonumber \\
    \sqrt{\e}(\sqrt{\e} + (1+\e)k) & \leq \frac{2\delta}{\gamma} \mcom \nonumber\\
    \e & \leq \paren{\frac{2\delta k }{\gamma (1 + k)}}^2 & \paren{\text{Since }\e \leq \frac{1}{4k^2}}   \label{eq:eps-val} \mper
  \end{align}
  Therefore setting $\e \leq \paren{\frac{2\delta k }{\gamma (1 + k)}}^2$ we get statement of the proposition.
\paragraph{Success Probability:} From the above analysis, we know that only in the \textbf{case (a)} in the analysis of \prettyref{lem:closepts} needs sampling. We took a sample of size $s = \frac{8k^3}{\e^9} \ln\frac{k^2}{\e^6}$. With probability $1-\frac{\e}{k}$, it contains $\frac{k}{\e^5}$ points from $P_j$.
Meanwhile, with probability, $1-\frac{\e}{k}$, $\norm{\pi-\tau_j}^2 \leq \e^4a^2$. Hence the success probability in the $j^{th}$ iteration is $\paren{1-\frac{\e}{k}}^2$. Therefore the success probability in $k$ iterations is $\paren{1-\frac{\e}{k}}^{2k} \geq 1-2\e$.

\paragraph{Runtime Analysis:} Each node in the returned tree by \prettyref{alg:tree} has $\abs{\cR} 2^{s+j} \paren{\frac{32j}{\e^2}}^j$ children, where $\abs{\cR} = \bigo{1/\e}$, and $s = \frac{8k^3}{\e^9} \ln \frac{k^2}{\e^6}$. Since the tree has height $k$, the number of candidate points for the means are  $\bigo{2^{\poly\paren{\frac{k}{\e}}}}$.
Since each node takes $\bigo{\abs{\cR} 2^{s+j} \paren{\frac{32j}{\e^2}}^j nd}$
time, the time complexity of the \prettyref{alg:tree} is $\bigo{2^{\poly\paren{\frac{k}{\e}}}nd}$.
The \prettyref{alg:tree} is called by the \prettyref{alg:centers} $\bigo{\log_{(1+\e)}\paren{\frac{\alpha+1}{\alpha-1}}}$ times.  Therefore, the total running time of the \prettyref{alg:centers} is $ \bigo{2^{\poly\paren{\frac{k}{\e}}}nd} $.
\end{proof}


\subsubsection{Proof of Claims} \label{sec:claim-pf}

\begin{proof}[Proof of \prettyref{claim:radius}]
  We know that
  \[ 2^{-1/2} \sqrt{\e} \gamma \radm{X} \leq \sqrt{\e} \gamma \radm{X} \leq  2^{1/2} \sqrt{\e} \gamma \radm{X} \mper \]
  Together with the fact that $\zeta \in \left [\radm{X}, (1+\e)\radm{X} \right]$, we get that
  \[ \frac{\sqrt{2}}{2} \sqrt{\e} \frac{\zeta \gamma}{(1+\e)} \leq \sqrt{\e} \gamma \radm{X} \leq \sqrt{2}\sqrt{\e} ~ \zeta \gamma \mper \]
  Let $\Hat r_j = \sqrt{2}\sqrt{\e}~\zeta \gamma$, we get that
  \[ \sqrt{\e} \gamma \radm{X} \leq \Hat r_j \leq 2(1+\e) \sqrt{\e} \gamma \radm{X} \mper\]
  Let $z = \frac{j ~ \Hat r_j}{j\sqrt{\e} \gamma \radm{X} }$. Then we have $1 \leq z \leq 2(1+\e)$. We build a grid in the interval $\Brac{\frac{z}{2(1+\e)},z}$ with grid length $\frac{\e}{4(1+\e)}z$, and obtain a number set $\cN = \Set{\frac{1+l \frac{\e}{2}}{2(1+\e)}z \middle| 0 \leq l \leq 4+\frac{2}{\e}}$.
  We prove that there must exist one number in $\cN$ and is between $1$ and $1+\e/2$. First, we know that $\frac{z}{2(1+\e)} \leq 1 \leq z$. If $z \leq 1 + \e/2$, we find the desired number in $\cN$.
  Otherwise, the whole interval $\Brac{1, 1+\e/2}$ is inside $\Brac{\frac{z}{2(1+\e)},z}$. Since each grid has length $\frac{\e}{4(1+\e)}z \leq \frac{\e}{4(1+\e)}2(1+\e) = \e/2$, there must exist one grid point located inside $\Brac{1,1+\e/2}$. Thus the desired number exists in $\cN$.

  $\cR = \Set{\frac{1+l \frac{\e}{2}}{2(1+\e)}~j~\Hat r_j \middle| 0 \leq l \leq 4+\frac{2}{\e}}$, and from the above analysis we know that there exists one value $r_j \in \cR$ such that
  \[ j\sqrt{\e} \gamma \radm{X} \leq r_j \leq \paren{1+\frac{\e}{2}}j\sqrt{\e} \gamma \radm{X} \mper \]
  Note that
\end{proof}

\begin{proof}[Proof of \prettyref{claim:peeling}]
  For each $1 \leq l \leq j-1$, we have that $\abs{\Opt_l \setminus \paren{\cup_{w=1}^{j-1} \cB_{j,w}}} \leq \abs{\Opt_l \setminus \cB_{j,l}}$. By Markov's inequality we have
  \[ \abs{\Opt_l \setminus \cB_{j,l}} \leq \frac{\sigma_l^2}{\paren{r_j - \norm{p_{v_l} - \mu_l}}^2} \abs{\Opt_l} \mper \]
  The above inequality is true because suppose we take a random variable $Y$ to be $\norm{x - \mu}$. We know that $\E{Y} = \sigma_l^2$. Therefore using Markov's inequality for some $r_0>0$ we get that
  \[ \Pr{Y \geq r_0} \leq \frac{\sigma_l^2}{r_0}  \mper \]
  Multiplying the above equation by $\abs{C_l}$, we get that the expected number of points outside the radius $r_0$ from center is less than  $\sigma_l^2 \mod{C_l} / r_0$.
  Note that $\sigma_{l}^2 \leq \gamma^2 \radm{X}^2$.
  Together with $r_j \geq j\sqrt{\e} \gamma \radm{X}$ and $\norm{p_{v_l} - \mu_l} \leq \e \gamma \radm{X} + (1+\e) l \sqrt{\e} \gamma \radm{X}$, we get
  \begin{align*}
    r_j - \norm{p_{v_l} - m_l} & \geq j\sqrt{\e} \gamma \radm{X} - \paren{\e \gamma \radm{X} + (j-1) (1+\e) \sqrt{\e} \gamma \radm{X}} \mcom \\
    & = (1-(j-1)\e - \sqrt{\e}) \sqrt{\e} \gamma \radm{X} \mper
  \end{align*}
  Thus we have
  \begin{align*}
    \abs{\Opt_l \setminus \cB_{j,l}} & \leq \frac{\sigma_l^2}{(1-(j-1)\e - \sqrt{\e})^2 \e \gamma^2 \radm{X}^2 } \abs{\Opt_l} \mcom\\
    & \leq \frac{\abs{\Opt_l}}{(1-(j-1)\e - \sqrt{\e})^2 \e} \mcom & \paren{\sigma_{l}^2 \leq \gamma^2 \radm{X}^2} \mcom\\
    & \leq \frac{\lambda_j ~ n}{(1-j\sqrt{\e})^2 \e} \mper
  \end{align*}
  Note that we can assume $\e$ is small enough such that $\e \leq \frac{1}{4k^2}$, which implies $\frac{\lambda_j n}{(1-j\sqrt{\e})^2 \e} \leq \frac{4 \lambda_j n}{\e}$, since $k \geq j$. Otherwise, we can just replace $\e$ by $\frac{\e}{4k^2}$ as part of input at the beginning of the algorithm. Thus we have that
  \begin{equation*}
    \abs{\Opt_l \setminus \cB_{j,l}} \leq \frac{4 \lambda_j n}{\e} \mper
  \end{equation*}
\end{proof}

\section{Lower Bound}
\subsection{Hardness Result}
Our hardness result immediately follows from Awasthi et al. \cite{AwasthiCKS15}. We will show a reduction from Vertex-Cover problem to the $\alpha$-center proximal $k$-means clustering with balanced clusters. The Vertex-Cover problem can be stated as follows: Given an undirected graph $G=(V,E)$, choose a subset $S$ of vertices with minimum $\abs{S}$, such that $S$ is incident on every edge of the graph. Awasthi et al. \cite{AwasthiCKS15} showed the following lemma:
\begin{lemma}[Corollary 5.3, \cite{AwasthiCKS15}]\label{lem:vertex-cover-hard}
  Given any unweighted triangle-free graph $G$ with bounded degrees, it is NP-hard to approximate Vertex-Cover within any factor smaller than $1.36$.
\end{lemma}

\begin{theorem}\label{thm:reduction}
  There exists constants $\alpha > 1$, $\omega >0$, $\e >0$, such that there is an efficient reduction from instances of Vertex-Cover on triangle-free graphs of bounded degree to those of $\alpha$-center proximal instances of Euclidean $k$-means clustering, where the size of each cluster is at least $\omega n / k$, that satisfies the following properties:
  \begin{enumerate}
    \item[(i)] if the Vertex-Cover instance has value $k$, the optimal $\alpha$-center proximal $k$-means clustering where the size of each cluster is at least $\omega n / k$, has cost at most $m - k$.
    \item[(ii)] if the Vertex-Cover instance has value at least $k(1 + \e)$, then the optimal $\alpha$-center proximal $k$-means clustering, where the size of each cluster is at least $\omega n / k$, has a cost at least $m - (1 - \Omega(\e))k$.
  \end{enumerate}
\end{theorem}

\begin{proof}
  The construction of the $k$-means instance is same as that of \cite{AwasthiCKS15}. Let $G=(V,E)$ denote the graph in the Vertex Cover instance $\cI$, with parameter $k$ denoting the number of vertices we can select. We assume that the graph $G$ is triangle free and with a maximum degree $\Delta = \Omega(1)$.
  Let $n$ be the number of vertices in the graph and $m$ be the number of edges.
  We construct the $k$-means instance $\cI_{km}$ as follows:
  for each vertex $i \in [n]$, we have a unit vector $x_i = (0,\ldots,0,1,0,\ldots,0)$ which has $1$ in the $i^{th}$ coordinate and $0$ elsewhere. For each edge $e \equiv (i,j)$, we have a vector $x_e = x_i+x_j$. Our data points are $\set{x_e : e \in E}$.

  \paragraph{Completeness:} suppose $\cI$ is such that there exists a vertex cover $S^* = \set{v_1,\ldots,v_k}$ of $k$ vertices. We will show that we can recover an $\alpha$-center proximal $k$-means clustering (where the size of each cluster is at least $\omega n / k$), of low cost.

  Let $E_{v_l}$ denote the set of edges covered by $v_l$ for $1 \leq l \leq k$. If an edge is covered by two vertices, we will assume that only one of them covers it (arbitrarily). As a result each $E_{v_l}$ is pairwise disjoint and their union is $E$.

  We now do the clustering as follows. Consider a cluster $\cF_v \defeq \set{x_e : e \in E_v}$, which consists of data-points associated with edges covered by a single vertex $v$. Let $m_{\cF_v}$ denote the number of edges $v$ cover in the vertex cover, and let $\mu_{\cF_v}$ denote the mean of $\cF_v$.
  The mean $\mu_{\cF_v}$ has a $1$ in one of the coordinates (corresponding to $x_v$), $\paren{1/\abs{m_{\cF_v}}}$ in $m_{\cF_v}$ coordinates (corresponding to the edges), and $0$ in the remaining.

  \begin{Claim}\label{claim:completeness}
  There exists an $\alpha$-center proximal $k$-means clustering (where the size of each cluster is at least $\omega n / k$) of $\cI_{km}$ with cost at most $m-k$, where $m$ is the number of edges in the graph $G=(V,E)$ associated with the vertex cover instance $\cI$, and $k$ is the size of the optimal vertex cover.
  \end{Claim}
  \begin{proof}
  The cost of the cluster $\cF_v$ is $\sum_{x \in \cF_v} \norm{x - \mu_{\cF_v}}^2$. We note that for any $x \neq x' \in \cF_v$, $\norm{x - \mu_{\cF_v}}^2 = \norm{x' - \mu_{\cF_v}}^2$. Therefore we get that the cost of the cluster $\cF_v$ is:
  \[ m_{\cF_v} \paren{ (m_{\cF_v} - 1) \paren{\frac{1}{m_{\cF_v}}}^2 + \paren{1-\frac{1}{m_{\cF_v}} }^2 } = m_{\cF_v} -1 \mper\]
  Summing this over the $k$ cluster gives us the cost $m-k$.

  Next, we bound the value of alpha for which this cluster is alpha stable. We note that the points closest to some other cluster center is the edge which was covered by two vertices. Let $x_e$ be covered by $v_{1}$ and $v_{2}$. W.l.o.g. we assume that$x_e \in \cF_{v_{1}}$.
  Let the number of edges in the cluster $\cF_{v_{1}}$ be $m_{v_1}$ and $\cF_{v_{2}}$ be $m_{v_2}$, and let their respective means be  $\mu_{v_1}$ and $\mu_{v_2}$.
  The distance of $x_e$ to mean of $\cF_{v_{1}}$ is
  \[ \norm{x_e-\mu_{v_1} }^2 = (m_{v_1} - 1) \paren{\frac{1}{m_{v_1}}}^2 + \paren{1-\frac{1}{m_{v_1}} }^2  = \frac{m_{v_1} -1}{m_{v_1}} \mcom \]
  and the distance of $x_e$ to mean of $\cF_{v_{2}}$ is
  \[ \norm{x_e-\mu_{v_2} }^2 = (m_{v_2}) \paren{\frac{1}{m_{v_2}}}^2 + 1 = \frac{1+m_{v_2}}{m_{v_2}} \mper \]
  Note that the maximum degree of the graph is $\Delta$. Therefore the value of $\alpha$ is
  \[ \alpha = \sqrt{ \frac{m_{v_1}(m_{v_2} + 1)}{m_{v_2}(m_{v_1} - 1)} } \geq \sqrt{ \frac{\Delta + 1}{\Delta -1} } \mper \]
  In the case where clusters do not share an edge, say $x_e \in \cF_{v_1}$ we get that that (as per previous calculation)
  \[ \norm{x_e-\mu_{v_1}}^2 = \frac{m_{v_1} -1}{m_{v_1}} \mcom \]
  and the distance of $x_e$ to $\mu_{v_2}$ is
  \[ \norm{x_e-\mu_{v_1}}^2 = 3 + (m_{v_2}) \paren{\frac{1}{m_{v_2}}}^2 \]
  Therefore, we get that the value of $\alpha$ in this case is
  \[ \alpha = \sqrt{ \frac{ 3+\frac{1}{m_{v_2}} }{ 1 - \frac{1}{m_{v_1}} } } \geq \sqrt{ \frac{\Delta + 3}{\Delta-1} } > \sqrt{ \frac{\Delta+1}{\Delta-1} } \mper \]
  Next we bound the value of $\omega$. The size of a cluster is bounded by the degree of the graph, ie., $\Delta$, and the minimum size of a vertex cover for a graph with bounded degree $\Delta$ is $\abs{S^*} \geq n / \Delta$.   Therefore we get that the value of $\omega$ is $\frac{\omega n}{k} \geq 1$, $\omega \geq \frac{k}{n} \geq \frac{1}{\Delta}$.
  \end{proof}

  \paragraph{Soundness:} Next we show that if there is an $\alpha$-center proximal $k$-means clustering, where the size of each cluster is at least $\omega n / k$, which has a low $k$-means cost, then there is a very good vertex cover for the corresponding graph. The proof for the soundness follows directly from Theorem 4.7 of \cite{AwasthiCKS15}.

  \begin{lemma}[Theorem 4.7, \cite{AwasthiCKS15}]\label{lem:soundness}
  If the $k$-means instance $\cI_{km}$ has a clustering $\Gamma = \set{\cF_1,\ldots,\cF_k}$ $\sum_{\cF \in \Gamma}Cost(\cF) \leq m-(1-\xi)k$, then there exists a $(1 + \bigo{\delta})k$-vertex cover of $G$ in the instance $\cI$.
  \end{lemma}

  The proof for our case follows from the above lemma because the statement holds for any $k$-means clustering of low cost, and hence it also holds for $\alpha$-center proximal $k$-means clustering, where the size of each cluster is at least $\omega n / k$.

  Combining \prettyref{claim:completeness} and \prettyref{lem:soundness} we get the proof of \prettyref{thm:reduction}.
\end{proof}

\begin{proof}[Proof of \prettyref{thm:hard}]
From \prettyref{thm:reduction} we get that for some constant $\alpha>1$, $\omega >0$, $\e > 0$ if the vertex cover has value $k \geq m/ \Delta$, then the $\alpha$-center proximal $k$-means clustering, where the size of each cluster is at least $\omega n / k$, has cost at most $m \paren{1-\frac{1}{\Delta}}$, and if the vertex cover is at least $k(1+\e)$,
then optimal $\alpha$-center proximal $k$ means cost is at least $m\paren{1-\frac{1-\Omega(\e)}{\Delta}}$.
The vertex cover hardness \prettyref{lem:vertex-cover-hard} says that it is NP-hard to distinguish if the resulting $\alpha$-center proximal $k$-means clustering, where the size of each cluster is at least $\omega n / k$, has cost at most $m\paren{1-\frac{1}{\Delta}}$ or cost more than $m\paren{1-\frac{1-\Omega(\e)}{\Delta}}$.
Since $\Delta$ is a constant, this implies that it is NP-hard to approximate $\alpha$-center proximal $k$-means problem to within some factor $(1+\Omega(\e))$, thereby proving the \prettyref{thm:hard}
\end{proof}

\subsection{On the Size of Possible Clustering}

\begin{proof}[Proof of \prettyref{prop:size-alpha-opt}]
  Our construction is similar to the instance constructed by \cite{BhattacharyaJK18} in their Theorem 2. We first construct the set of points $X$ for which we fix an integer $m$ such that $\alpha = \sqrt{ 1 + \frac{2}{m-1} } \leq \alpha'$. From this we get that $ m = \frac{2}{\alpha^2-1} + 1$.
  The points in $X$ belongs to $\R^d$, where $d = km$. The set $X$ will have $n = d$ points : the standard basis vectors of $\R^d$, denoted by $e_1, \ldots,e_d$.
  Now, we define the set of $\alpha$-center proximal clusterings $\bbC$.
  The set $\bbC$ will consist of clusterings $C = \set{C_1,\ldots,C_k}$, for which each of the clusters has exactly $m$ points. We will now show that such a clustering is $\alpha$-center proximal. Consider any two clusters, say $C_1$ and $C_2$. The mean $\mu_1$ of $C_1$ has the value $1/m$ in its $m$ coordinates and $0$ in other coordinates.
  The  $\mu_2$ of $C_2$ has the value $1/m$ in its $m$ coordinates and $0$ in other coordinates, and for a non-zero coordinate of $\mu_1$, $\mu_2$ has $0$ in the respective coordinate, since the clusters are disjoint by definition. Therefore, consider a point $x_i \in C_1$, we get that
  \[ \norm{x_1 - \mu_1}^2  = \paren{1-\frac{1}{m}}^2 + (m-1)\paren{\frac{1}{m}}^2 = 1-\frac{1}{m} \mper\]
  and
  \[ \norm{x_1 - \mu_2}^2 = 1+m\paren{\frac{1}{m}}^2 = 1+\frac{1}{m} \mper \]
  Therefore the value of $\alpha$ is
  \[ \alpha = \sqrt{\frac{1+\frac{1}{m}}{1-\frac{1}{m}}} = \sqrt{\frac{m+1}{m-1}} \mper \]
  The number of such possible clusterings are
  \[ \abs{\bbC} = \frac{(km)!}{(m!)^k} \approx k^{(km)}\mper \]
  Therefore we get that the total possible number of such clusterings are
  \[ k^{k \paren{ \frac{2}{\alpha^2-1} } + 1} = 2^{ \tilde{\Omega}\paren{ \frac{k}{\alpha^2 - 1} }} \mper \]
\end{proof}

\subsection*{Acknowledgements}
AL is grateful to Microsoft Research for supporting this collaboration. AL was supported in part by SERB Award ECR/2017/003296. We thank Ravishankar Krishnaswamy for helpful pointers.

\bibliographystyle{amsalpha}
\bibliography{ref}
\end{document}